\documentclass[fleqn,11pt]{wlscirep}
\usepackage[utf8]{inputenc}
\usepackage[T1]{fontenc}
\usepackage{hyperref}
\usepackage{authblk}
\usepackage{xcolor}
\usepackage{ragged2e}
\usepackage{algorithm}
\usepackage{algpseudocode}
\usepackage[normalem]{ulem}

\newcommand{\name}{JANUS}

\definecolor{R0}{rgb}{0,0,1}

\title{JANUS: A Multimodal Foundation Neural Sampler for Disordered Materials}

\author[1,2,*,$\#$]{Denis Blessing}
\author[1,3,4,*,$\#$]{Mouyang Cheng}
\author[5]{Maximilian Schebek}
\author[6]{Jutta Rogal}
\author[3,7]{Mingda Li}
\author[1,$\dagger$]{Carles Domingo-Enrich}
\author[1,$\dagger$]{Yuanqi Du}

\affil[1]{Microsoft Research New England, Cambridge, MA 02142, USA}
\affil[2]{Karlsruhe Institute of Technology, 76131 Karlsruhe, Germany}
\affil[3]{Center for Computational Science and Engineering, MIT, Cambridge, MA 02139, USA}
\affil[4]{Department of Materials Science and Engineering, MIT, Cambridge, MA 02139, USA}
\affil[5]{Department of Physics, Freie Universit\"{a}t Berlin, 14195 Berlin, Germany}
\affil[6]{Initiative for Computational Catalysis, Flatiron Institute, New York, NY 10010, USA}
\affil[7]{Department of Nuclear Science and Engineering, MIT, Cambridge, MA 02139, USA}

\affil[*]{These authors contributed equally.}
\affil[$\#$]{Work done during internship at Microsoft Research.}
\affil[$\dagger$]{Correspondence. Email: carlesd@microsoft.com, yuanqidu@microsoft.com}

\begin{abstract}
Disordered materials are central to modern technologies, from high-performance alloys to quantum devices, yet understanding and designing them requires exploring an enormous space of compositions and atomic arrangements. Here we introduce JANUS, a multimodal neural sampler that jointly explores chemical composition, atomic structure and volume, learning directly from the energy of the system without requiring pre-generated equilibrium configurations. Across test systems, JANUS reproduces reference equilibrium properties and recovers free energy and phase behavior with about three orders of magnitude fewer energy evaluations than reference Monte Carlo sampling. In multicomponent alloys, JANUS reproduces experimentally observed local chemical order and can be steered to further enhance chemical order and stiffness. Coupled to an evolutionary language model agent, it also enables efficient design of alloys with balanced optical and mechanical properties. Beyond alloys, JANUS explores vacancies and dopants spanning 15 elements in semiconductors including silicon and diamond, recovers established defects including the silicon $E$ centre, and identifies new candidate defect pairs and triplets for quantum engineering, including S-Ti in silicon and B-[O]-O in diamond, with deep in-gap states validated by hybrid-functional density functional theory. By unifying exploration across composition, atomic arrangement and volume, JANUS opens a new avenue for understanding, characterizing and designing chemically disordered materials.
\end{abstract}

\begin{document}
\addtocontents{toc}{\protect\setcounter{tocdepth}{-10}}

\flushbottom
\maketitle
\thispagestyle{empty}

\section*{Introduction}
The thermodynamic behavior of a chemically disordered crystal emerges from an ensemble of chemical configurations, atomic displacements, defects and cell volumes \cite{simonov2020designing,peng2026atomistic,esters2021settling,esters2023qh}. Chemical disorder arises when multiple elements or vacancies occupy common crystallographic sites without forming a perfectly ordered arrangement, as in substitutional alloys, cation-disordered ceramics and defect-rich semiconductors \cite{simonov2020designing,peng2026atomistic}. Competition between energetic interactions and configurational entropy can produce clustering and chemical short-range order \cite{chen2021direct,wu2023chemical}, while thermal atomic motion and cell fluctuations alter local coordination, strain and bonding \cite{song2017local,esters2021settling,esters2023qh}. These coupled chemical and structural fluctuations shape the behavior of materials ranging from high-entropy alloys and ceramics to disordered solid electrolytes \cite{george2019high,oses2020high,zeng2022high}, influencing phase stability \cite{esters2021settling,kam2023crystal}, ionic transport \cite{ji2019hidden}, catalytic activity \cite{mints2026unraveling}, mechanical response \cite{chen2021simultaneously} and defect energetics \cite{baldassarri2023oxygen}. Predictive modeling therefore requires a joint description of chemical configurations, lattice vibrations and cell relaxation.

Obtaining this ensemble description remains computationally demanding even when only occupational disorder is considered. Markov chain Monte Carlo explores chemical configurations through sequential moves, requiring repeated energy evaluations. Slow mixing near phase transitions and the growing cost of large supercells make systematic exploration across temperatures and compositions particularly expensive \cite{widom2018modeling,peng2026atomistic}. Much of this sampling effort must also be repeated for each new thermodynamic condition. Accounting for structural relaxation and vibrational motion adds another layer of complexity: hybrid approaches couple chemical transmutation moves with molecular dynamics, substantially increasing the number of energy and force evaluations \cite{sadigh2012scalable,widom2014hybrid}. These costs limit both the characterization of equilibrium behavior and the use of thermodynamic sampling within materials design.

Neural samplers have been proposed as a potential route to addressing this bottleneck by learning generative processes that approximate target Boltzmann distributions. Pioneered by Boltzmann generators \cite{noe2019boltzmann}, these approaches aim to produce effectively independent configurations. For chemical disorder, recent efforts have explored autoregressive models with tractable likelihoods \cite{wu2019solving,damewood2022sampling,du2026scaling} and discrete diffusion models over categorical configurations \cite{holderriethleaps2025,zhu2025mdns,guo2026discrete}. However, these attempts have largely been restricted to chemical occupations on fixed lattice sites and cell geometries, leaving the coupled sampling problem in realistic disordered crystals unresolved. Addressing this challenge requires jointly capturing chemical identities, vacancies, atomic motion, and changes in volume.

Here we introduce JANUS, a multimodal neural sampler that jointly treats these discrete and continuous degrees of freedom. Building on continuous neural samplers in the $NVT$ ensemble \cite{cheng2026atlas}, JANUS couples discrete diffusion over atomic identities and vacancies with continuous diffusion over atomic displacements and cell volume. It learns directly from the target potential energy using configurations generated by the sampler itself, without requiring pre-generated equilibrium data.  This construction supports on-lattice isobaric semi-grand-canonical $\Delta\mu NPT$ and grand-canonical $\mu VT$ sampling with a single model that is reusable across temperatures and chemical potentials and transferable across system sizes. Path weights associated with the generative process enable free-energy and ensemble-property estimation as well as conditional steering, establishing the trained neural sampler for quantitative thermodynamic predictions and inverse design.

We first demonstrate the capability of JANUS on the Ising model and Cu--Ni alloys, where JANUS recovers equilibrium ensembles when chemical and structural degrees of freedom become coupled. We then examine alloys with phase separation and competing crystal structures, where free energy differences determine which states are thermodynamically favored. JANUS recovers the corresponding free energy landscapes and phase boundaries with about three orders of magnitude fewer energy evaluations than reference Monte Carlo algorithms, even when the one-time cost of training the amortized sampler is included. This ability to characterize equilibrium disorder also provides a starting point for targeted design. In CrCoNi, JANUS reproduces experimentally observed local chemical order and can be steered toward selected ordering patterns or greater stiffness. Coupling the sampler to an evolutionary language model agent extends this exploration across alloy compositions, identifying candidates that balance optical and mechanical properties.

The same framework opens a route to exploring defect combinations in semiconductors, where the identities and arrangements of a few atoms can create distinct electronic environments. In silicon and diamond, JANUS jointly samples vacancies, dopants and local structural relaxation, recovering established defects while proposing new pairs and higher-order complexes. Hybrid-functional density functional theory (DFT) calculations identify candidates with in-gap electronic levels and spin-polarized states, motivating further investigation for quantum engineering. Together, these applications show how joint sampling of chemical and structural disorder can connect equilibrium understanding with the search for materials and defects with desired properties.

\section*{Results}

\begin{figure}[!h]
  \centering
  \includegraphics[width=\textwidth]{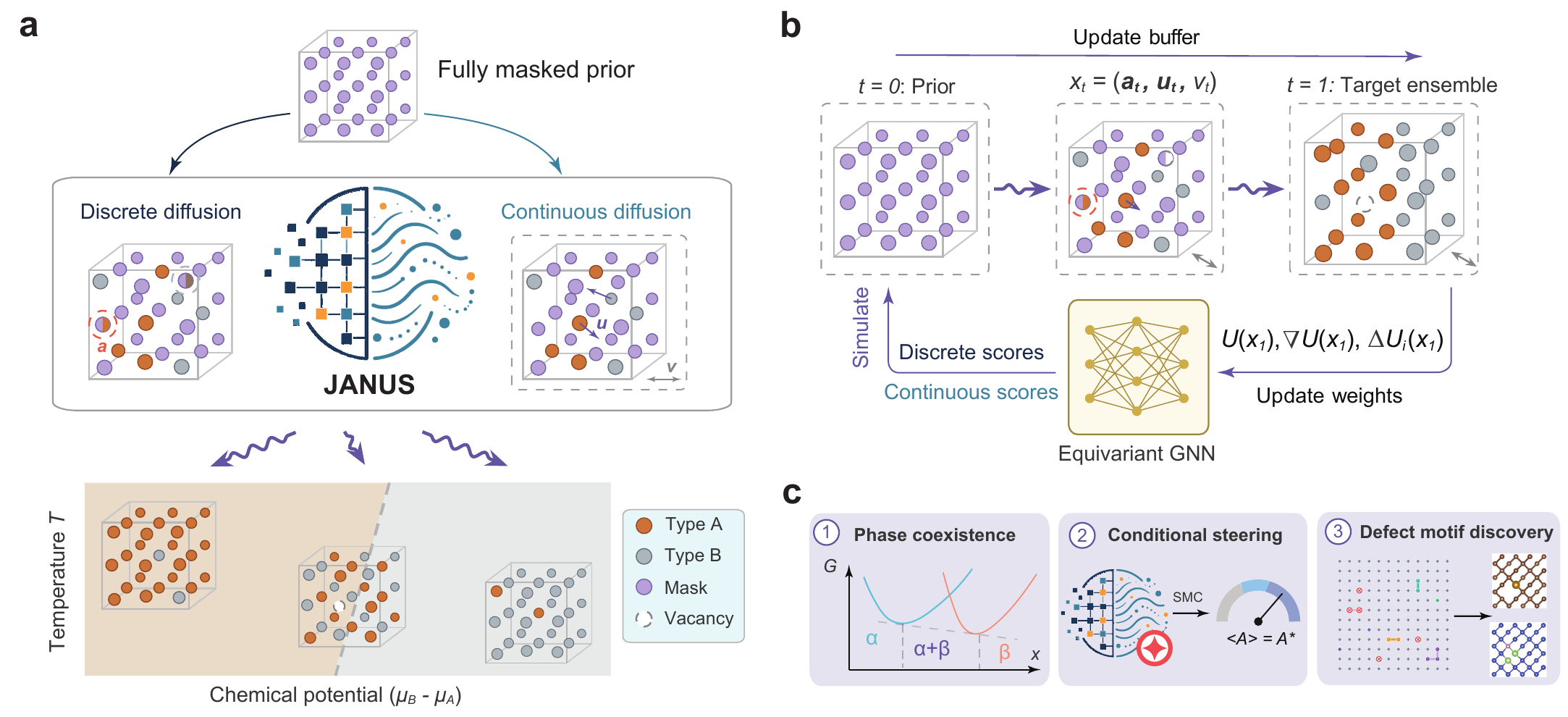}
  \caption{\textbf{Overview of the JANUS framework.} \textbf{a.} JANUS couples masked discrete diffusion over atomic identities and vacancies with continuous diffusion over atomic displacements and cell volume, enabling joint sampling of diverse disordered materials configurations across thermodynamic conditions. 
  \textbf{b.} Self-bootstrapped training framework of JANUS. Intermediate states $x_t=(\mathbf{a}_t, \mathbf{u}_t, v_t)$ are generated between the prior ($t=0$) and target ensemble ($t=1$) from stochastic processes, and an $E(3)$-equivariant graph neural network learns the discrete and continuous scores using the potential energy $U(x_1)$, force $-\nabla U(x_1)$ and discrete changes $\Delta U_i(x_1)$ of the generated structures at $t=1$. The training buffer is iteratively updated with newly generated configurations. 
  \textbf{c.} Downstream applications enabled by the trained JANUS sampler, including determination of phase coexistence from thermodynamic free energy, inference-time steering towards target properties and systematic discovery of novel defect motifs.}
  \label{fig1}
\end{figure}

\subsection*{Overview of JANUS}
The goal of JANUS is to sample directly from a thermodynamic distribution specified by an energy function $U$, without requiring pre-generated equilibrium configurations. We consider substitutional alloys and point defects in crystalline hosts, represented by $N$ reference lattice sites at fractional positions $\mathbf q_i$ within a cubic supercell. A configuration $(\mathbf a,\mathbf x)=(\mathbf a,\mathbf u,v)$ comprises the site occupations $\mathbf a$, fractional atomic displacements $\mathbf u$, and logarithmic supercell volume $v=\log V$. For occupied sites, the Cartesian positions are $\mathbf r_i=V^{1/3}(\mathbf q_i+\mathbf u_i)$. Sampling is restricted to the crystalline basin of the reference lattice, with chemical rearrangements represented through changes in $\mathbf a$.
Each site carries a token from an alphabet $\mathcal A$ containing the allowed atomic species and, where applicable, a vacancy token $\texttt{VAC}$. Defining $N_s(\mathbf a)=\sum_i\mathbf 1[a_i=s]$ for the number of sites carrying token $s$, the target density can be written as (see more details in Methods and Supplementary Information)
\begin{equation}
\label{eq:target}
p_{\text{target}}(\mathbf a,\mathbf x)\propto\exp\left[-\frac{1}{k_B T}\left(
U(\mathbf a,\mathbf x)+PV-\sum_{s\in\mathcal A}\mu_sN_s(\mathbf a)\right)\right].
\end{equation}

This formulation naturally covers a broad class of lattice-based disordered materials. For alloys, JANUS samples the isobaric semi-grand-canonical (SGC) $\Delta\mu NPT$ ensemble, allowing chemical populations, atomic displacements and volume to fluctuate at a fixed total number of atoms. For defects, it samples the grand-canonical (GC) $\mu VT$ ensemble at fixed volume, with a vacancy token allowing the number of physical atoms to vary while the number of lattice sites remains fixed. These settings activate the channels $(\mathbf a,\mathbf u,v)$ and $(\mathbf a,\mathbf u)$, respectively. JANUS generates configurations through coupled continuous and masked discrete diffusion (Fig.\,\ref{fig1}a). Over generation time $t\in[0,1]$, each active continuous channel follows stochastic-interpolant dynamics \cite{albergo2025stochastic},
\begin{equation}
\label{eq: fwd sde}
\mathrm d x_t^c=
\big[b_\theta^c+g^2(t)s_\theta^c\big](\mathbf a_t,\mathbf x_t,t)\,\mathrm dt
+\sqrt{2}\,g(t)\,\mathrm d w_t^c,
\qquad x_0^c\sim\pi_{0,c},
\qquad c\in\{u,v\}.
\end{equation}
Here, $b_\theta^c$ and $s_\theta^c$ are learned velocity and score fields conditioned on the full current state, $g(t)$ controls the noise amplitude, $w_t^c$ is a Wiener process, and $\pi_{0,c}$ is a Gaussian prior specified in Methods.

The species channel starts fully masked and progressively assigns chemical identities or vacancies using masked discrete diffusion \cite{austin2021structured,sahoo2024simple}. Each site is revealed once, at a random time governed by a schedule $\alpha_a(t)$, with its token drawn from a learned denoising posterior,
\begin{equation}
\label{eq: denoising posterior}
q_{\theta,i}(b\mid\mathbf a_t,\mathbf x_t)
\approx
\mathbb P\big(a_{1,i}=b\mid\mathbf a_t,\mathbf x_t\big).
\end{equation}
Revealed tokens remain fixed, while the continuous channels continue to evolve.

All learned quantities are predicted by a single $E(3)$-equivariant graph neural network (GNN) conditioned on the current configuration and thermodynamic conditions $(T,\Delta\mu)$, allowing one model to sample across temperatures and chemical potentials. JANUS is trained through self-consistent fixed-point bootstrapping \cite{blessing2026bridge,havens2026flow} (Fig.\,\ref{fig1}b). At each iteration, the sampler generates configurations whose energies, forces and single-site substitution energy changes provide training signals for the coupled continuous and discrete processes. Repeating this cycle refines the sampler using evaluations of the interatomic potential, without requiring pre-generated equilibrium data.

The trained sampler supports three representative applications (Fig.\,\ref{fig1}c). First, forward–backward path weights enable importance-weighted estimates of free energy and thermodynamic observables, providing access to phase coexistence and phase diagrams. Second, reward tilting steers sampling toward target properties \cite{he2026rne,rector2026general}, while integration with an evolutionary large language model (LLM) agent enables multi-objective design across alloy compositions. Third, amortized sampling across chemical space supports systematic discovery of structural motifs, including semiconductor defect complexes as candidates for quantum engineering.

\begin{figure}[!htbp]
  \centering
  \includegraphics[width=0.95\textwidth]{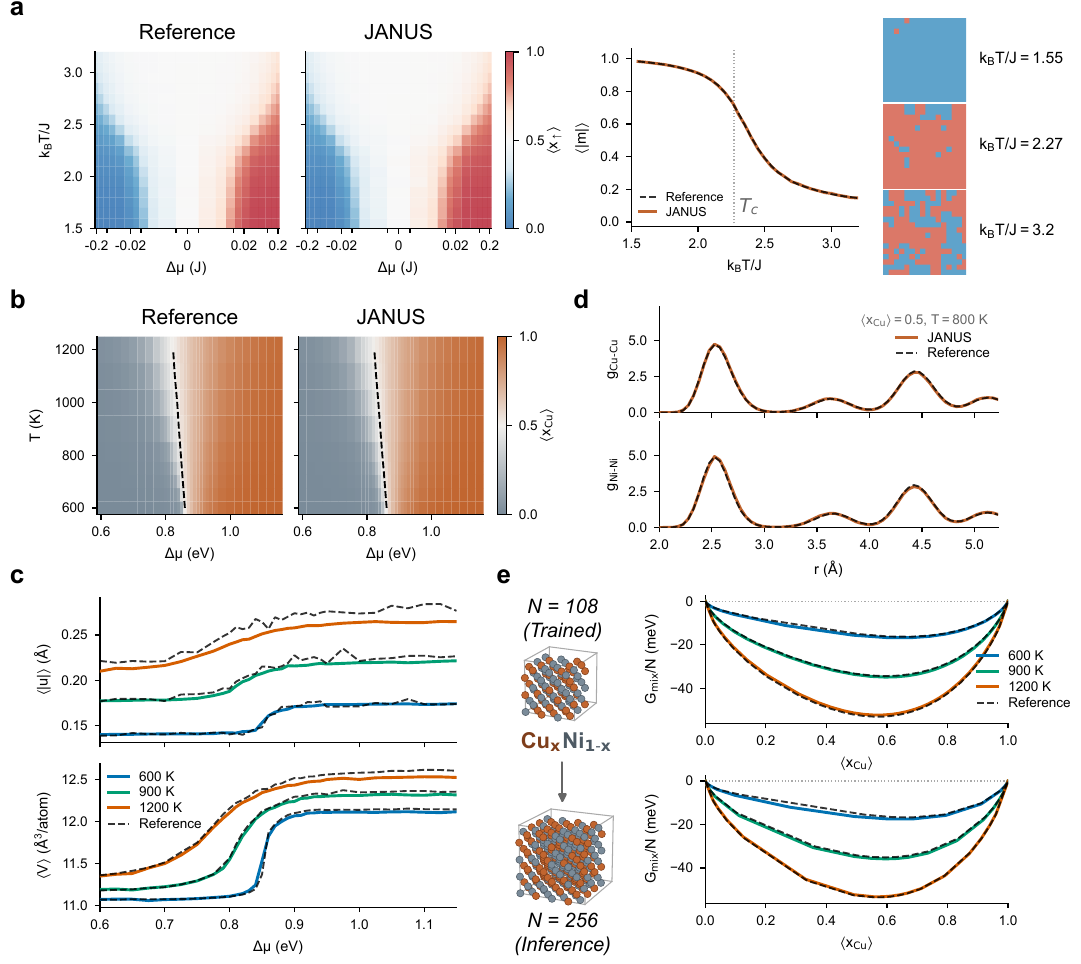}
  \caption{\textbf{Validation of JANUS across discrete and multimodal sampling.} 
  \textbf{a.} Comparison between reference and JANUS sampling for the two-dimensional ferromagnetic Ising model on a $16\times16$ lattice. Left: mean fraction of spin-up sites $\langle x_\uparrow\rangle$ is shown across temperature $T$ and external field $\Delta\mu=\mu_\uparrow-\mu_\downarrow$. Right: absolute magnetization $\langle |m| \rangle$ at $\Delta\mu=0$ as a function of temperature, with representative spin configurations shown at temperatures below, near and above $T_c\simeq2.27 J/k_B$.
  \textbf{b.} Validation on fcc Cu--Ni alloys with all three JANUS channels active. Average Cu concentration $\langle x_{\mathrm{Cu}}\rangle$ as a function of temperature and chemical-potential difference $\Delta\mu=\mu_\text{Cu}-\mu_\text{Ni}$, comparing reference and JANUS sampling. The dashed lines denote $\langle x_{\mathrm{Cu}}\rangle=0.5$. \textbf{c.} Mean atomic displacement $\langle|\mathbf{u}|\rangle$ and atomic volume $\langle V\rangle$ obtained by JANUS at $600$, $900$ and $1200$ K, with reference results shown by dashed lines.
  \textbf{d.} Partial radial distribution functions $g_{\mathrm{Cu-Cu}}(r)$ and $g_{\mathrm{Ni-Ni}}(r)$ of the equiatomic alloy ($\langle x_{\mathrm{Cu}}\rangle=0.5$) at $T=800$ K, compared with the reference ensemble.
  \textbf{e.} System size extrapolation from $N=108$ systems used during training to $N=256$ systems at inference, demonstrated by the composition-dependent mixing free energy $G_{\mathrm{mix}}/N$ at $600$, $900$ and $1200$ K. JANUS predictions are compared with reference values shown as dashed lines.}
  \label{fig2}
\end{figure}

\subsection*{Validation across discrete and multimodal sampling}
We first validate the novel algorithm developed for the discrete channel in JANUS using the prototypical two-dimensional ferromagnetic Ising model on a $16\times16$ lattice (Fig.\,\ref{fig2}a). In this setting, each lattice site carries only a discrete spin variable $\{+1, -1\}$, while $\Delta\mu$ plays the role of the external magnetic field and no continuous degrees of freedom are present.
Rather than training an individual sampler at each thermodynamic state, we condition a single JANUS model jointly on temperature $T$ and chemical potential difference $\Delta\mu$, and train it over the full two-dimensional $(T,\Delta\mu)$ domain. Once this one-time training is finished, the same model can directly generate equilibrium configurations at arbitrary thermodynamic conditions within the training range without state-specific retraining. 
Across this domain, JANUS closely reproduces the reference magnetization landscape obtained with Wolff cluster Monte Carlo\cite{wolff1989collective}. In particular, at $\Delta\mu=0$, the model captures the paramagnetic-to-ferromagnetic transition near the critical temperature $k_{\mathrm B}T_c/J\simeq2.27$, including the rapid emergence of finite magnetization below $T_c$. Representative spin configurations below, near and above $T_c$ further illustrate the evolution from long-range ferromagnetic order to disordered spins.

We next turn to a more realistic fcc Cu--Ni alloy, described by an embedded-atom method (EAM) potential \cite{fischer2019systematic}, where all three JANUS channels, chemical identities $\mathbf{a}$, atomic displacements $\mathbf{u}$ and cell volume $v$, are simultaneously active. As in the Ising benchmark, a single JANUS model is trained on a $3\times3\times3$ fcc supercell with 108 atoms, and amortized over the full $(T,\Delta\mu)$ space. All results below are subsequently generated from this amortized model. Reference ensembles are obtained from a classical Markov chain Monte Carlo (MCMC) algorithm that combines single-site species flips with atomic displacement and cell-volume moves under the same potential, run on a grid of $(T,\Delta\mu)$ states. 
JANUS reproduces the reference composition response, $\langle x_{\mathrm{Cu}}\rangle$, over a broad range of temperatures and chemical potential differences (Fig.\,\ref{fig2}b), including the temperature-dependent shift and broadening of the composition crossover. 
Beyond chemical identities, the generated ensembles also recover the corresponding structural and volumetric relaxation. As shown in Fig.\,\ref{fig2}c, the mean atomic displacement $\langle|\mathbf{u}|\rangle$ and equilibrium atomic volume $\langle V\rangle$ remain in close agreement with reference calculations across chemical potentials and temperatures, the largest deviation being a slight underestimation of $\langle|\mathbf{u}|\rangle$ at the highest temperature.
 In parallel, we examine the local atomic environment of the equiatomic alloy ($\langle x_{\mathrm{Cu}}\rangle=0.5$) at $T=800$ K using the partial radial distribution functions $g_{\mathrm{Cu-Cu}}(r)$ and $g_{\mathrm{Ni-Ni}}(r)$. As shown in Fig.\,\ref{fig2}d, the predictions from JANUS closely follow the reference ensemble over the first several coordination shells, indicating that the sampler accurately preserves the local structural correlations arising from the coupled chemical and displacement channels.
Finally, we demonstrate the size transferability of JANUS from a small training cell directly to larger cells, using the mixing Gibbs free energy as the evaluation metric, shown in Fig.\,\ref{fig2}e. For the $N=108$ system used during training, JANUS accurately reproduces the reference composition-dependent $G_\text{mix}$ across temperatures. Without retraining, the same model can be directly evaluated at $N=256$, where chemical disorder is better represented, and continues to reproduce $G_{\mathrm{mix}}$ across compositions and temperatures, demonstrating reliable system size extrapolation. This validation on the Cu--Ni alloy demonstrates that JANUS can accurately couple chemical, structural and volumetric fluctuations, amortize multimodal sampling across thermodynamic conditions, and transfer across system size.  A systematic comparison of the number of energy evaluations required by JANUS and by classical MCMC is given below for the Cu--Ag and Ni--Cr systems (Fig.\,\ref{fig3}c).

\begin{figure}[!htbp]
  \centering
  \includegraphics[width=0.95\textwidth]{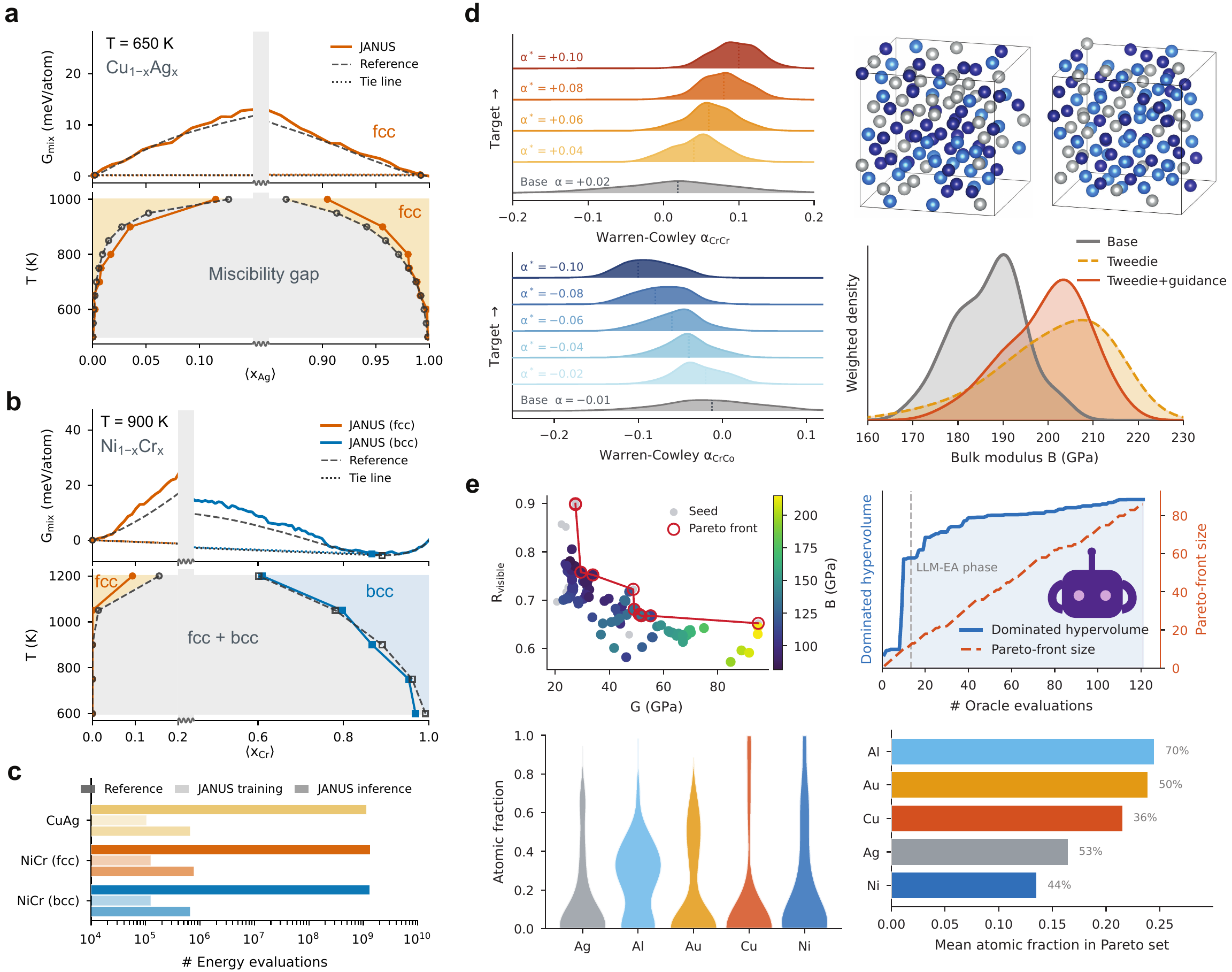}
  \caption{\textbf{Thermodynamic sampling, conditional generation and inverse design of alloys.} \textbf{a.}  Thermodynamic sampling of fcc Cu--Ag alloys. Composition-dependent mixing Gibbs free energy $G_\text{mix}(x)$ at $T=650$~K, its common-tangent construction (dotted line) and the resulting fcc miscibility-gap boundaries are shown from JANUS and reference calculations. \textbf{b.}  Thermodynamic sampling of Ni--Cr alloys in fcc and bcc phases. $G_\text{mix}(x)$ of both lattices at $T=900$~K, their common tangent (dotted line) and the resulting phase boundaries between Ni-rich fcc and Cr-rich bcc phases are shown for JANUS and reference calculations. 
  For each lattice family, a single JANUS model is amortized across the thermodynamic conditions considered.
  \textbf{c.} Number of target-energy evaluations required by the reference Monte Carlo calculations and by JANUS for Cu--Ag and for fcc and bcc Ni--Cr. For JANUS, the one-time training cost of each amortized model is shown separately from the inference cost. 
  \textbf{d.} Conditional steering in fcc CrCoNi alloys. Targeted distributions of the Warren-Cowley parameters $\alpha_{\mathrm{CrCr}}$ and $\alpha_{\mathrm{CrCo}}$ are obtained by tilting the base JANUS sampler with two independent rewards. Representative sampled supercells are shown ($\alpha_{\mathrm{CrCr}}=0.10$ and $\alpha_{\mathrm{CrCo}}=-0.10$, respectively). Reward tilting on the bulk modulus $B$ is also shown, using Tweedie-style reward estimates with or without gradient guidance. 
  \textbf{e.} Multi-objective inverse design over Ag--Al--Au--Cu--Ni alloys containing up to three elements using an LLM-driven evolutionary framework coupled to JANUS. The search jointly optimizes the bulk modulus $B$, shear modulus $G$, Pugh ratio $B/G$ and visible-range optical reflectance $R_{\mathrm{visible}}$. Shown are the Pareto front evaluated in the ($R_{\mathrm{visible}}$–$G$) subspace, the evolution of dominated hypervolume and Pareto front size, and the elemental fractions and occurrence statistics within the final Pareto set.}
  \label{fig3}
\end{figure}

\subsection*{Thermodynamic sampling and inverse design of alloys}
We next use JANUS to resolve alloy phase equilibrium and guide the design of local chemical order and material properties (Fig.\,\ref{fig3}). 
We first consider phase coexistence and boundaries of Cu--Ag alloys, which exhibit a pronounced fcc miscibility gap, and Ni--Cr alloys, where the Ni-rich fcc and Cr-rich bcc phases compete over a broad composition and temperature range, described by an EAM potential \cite{williams2006embedded} and a Finnis--Sinclair-type potential \cite{li2019strengthening}, respectively (Fig.\,\ref{fig3}a,b).
From JANUS-generated ensembles, we reconstruct composition-dependent mixing Gibbs free energy and determine phase coexistence boundaries (details in Methods). For Cu--Ag, JANUS recovers the free energy landscape underlying phase separation and the temperature-dependent miscibility gap (Fig.\,\ref{fig3}a). For Ni--Cr, comparing the separately sampled fcc and bcc branches yields phase boundaries in close agreement with fixed-composition parallel-tempering Monte Carlo reference calculations (Fig.\,\ref{fig3}b; Methods).
For each lattice, a single JANUS model is amortized across the temperatures and compositions considered (Methods); Ni--Cr requires two such models, one for each of its fcc and bcc phases. This amortization substantially reduces the computational cost of mapping phase behavior. For Cu--Ag and both Ni--Cr lattices, JANUS obtains comparable thermodynamic estimates with about three orders of magnitude  fewer energy evaluations than the reference MCMC calculations (Fig.\,\ref{fig3}c). This comparison includes both the one-time training cost of each amortized model and the inference cost of generating the ensembles underlying the free energy curves and binodals; the training cost alone lies four orders of magnitude below the reference cost and is shared by all thermodynamic conditions, so that the total JANUS cost is dominated by inference. 

JANUS also captures local chemical ordering within a disordered alloy, as demonstrated for medium-entropy fcc CrCoNi at $T=800$ K (Fig.\,\ref{fig3}d). Chemical short-range order (SRO) reflects preferences among neighboring species and can influence material properties through the resulting local chemical environments. We quantify these correlations using Warren–Cowley parameters $\alpha_{ij}$, which measure the tendency of species $i$ and $j$ to associate or avoid one another relative to random mixing. Without additional steering, the pretrained JANUS ensemble yields $\alpha_{\mathrm{CrCr}}=+0.02$ and $\alpha_{\mathrm{CrCo}}=-0.01$, indicating Cr--Cr avoidance and preferential Cr--Co association. The Cr--Cr avoidance agrees with experimental observations, while both ordering tendencies are consistent with first-principles predictions \cite{zhang2020short,sheriff2024quantifying}.

Beyond characterizing equilibrium chemical order, JANUS can steer the sampled configurations toward prescribed structural or mechanical properties. We bias the learned distribution through reward tilting,
\begin{equation}
p_{\eta}(\mathbf a,\mathbf x)
\propto
p_{\mathrm{base}}(\mathbf a,\mathbf x)
\exp[\eta r(\mathbf a,\mathbf x)],
\end{equation}
where $\eta>0$ controls the strength of the bias. For SRO, the reward penalizes deviations of a selected Warren–Cowley parameter from its target value. We implement the tilted sampling distribution using sequential Monte Carlo (SMC) at inference time (Methods). In separate steering experiments, the distributions shift toward stronger Cr--Cr avoidance, with $\alpha_{\mathrm{CrCr}}$ increasing from $+0.02$ to $+0.10$, and stronger Cr--Co affinity, with $\alpha_{\mathrm{CrCo}}$ decreasing from $-0.01$ to $-0.10$ (Fig.\,\ref{fig3}d). Thus, the sampler can both recover intrinsic ordering tendencies and generate ensembles with deliberately modified local correlations. The same approach extends to mechanical properties. Using a reward that favors a larger bulk modulus $B$, both reward-only and gradient-guided steering shift the generated ensemble toward stiffer configurations (Fig.\,\ref{fig3}d). This result shows that inference-time control can act on a functional property as well as on a structural descriptor, allowing the pretrained sampler to explore configurations with enhanced mechanical response without retraining.

To explore property trade-offs across alloy compositions, we couple JANUS to an evolutionary LLM agent (Fig.\,\ref{fig3}e). A single JANUS model samples fcc alloys in the Ag--Al--Au--Cu--Ni system at $T=900$ K. The agent proposes chemical-potential conditions, and the resulting equilibrium compositions and ensemble properties guide subsequent proposals. The search jointly optimizes the bulk and shear moduli $(B,G)$, the Pugh ratio $B/G$ as a ductility proxy, and visible-range optical reflectance $R_{\mathrm{visible}}$. Mechanical properties are evaluated using MACE-MPA-0 \cite{batatia2022mace,batatia2025foundation}, and reflectance using a DFT-trained surrogate (Methods).

The search rapidly expands the non-dominated region of predicted property space, with the hypervolume approaching a plateau after approximately 120 oracle evaluations, each corresponding to one JANUS ensemble (Fig.\,\ref{fig3}e). The resulting Pareto front spans trade-offs between stiffness, the ductility proxy and optical response across several composition families, providing alternative candidates for different property priorities.
The elemental statistics reveal systematic composition–property associations. Al occurs most frequently and contributes the largest mean fraction among Pareto candidates, particularly in the high-reflectance region, while Au contributes prominently despite appearing in fewer candidates. Ni is associated with higher stiffness, Ag contributes to both reflectance and the ductility proxy, and Cu occupies intermediate trade-off regions. These trends provide composition-level guidance for balancing optical and mechanical properties, illustrating how equilibrium sampling can inform alloy design.

\begin{figure}[!htbp]
  \centering
  \includegraphics[width=\textwidth]{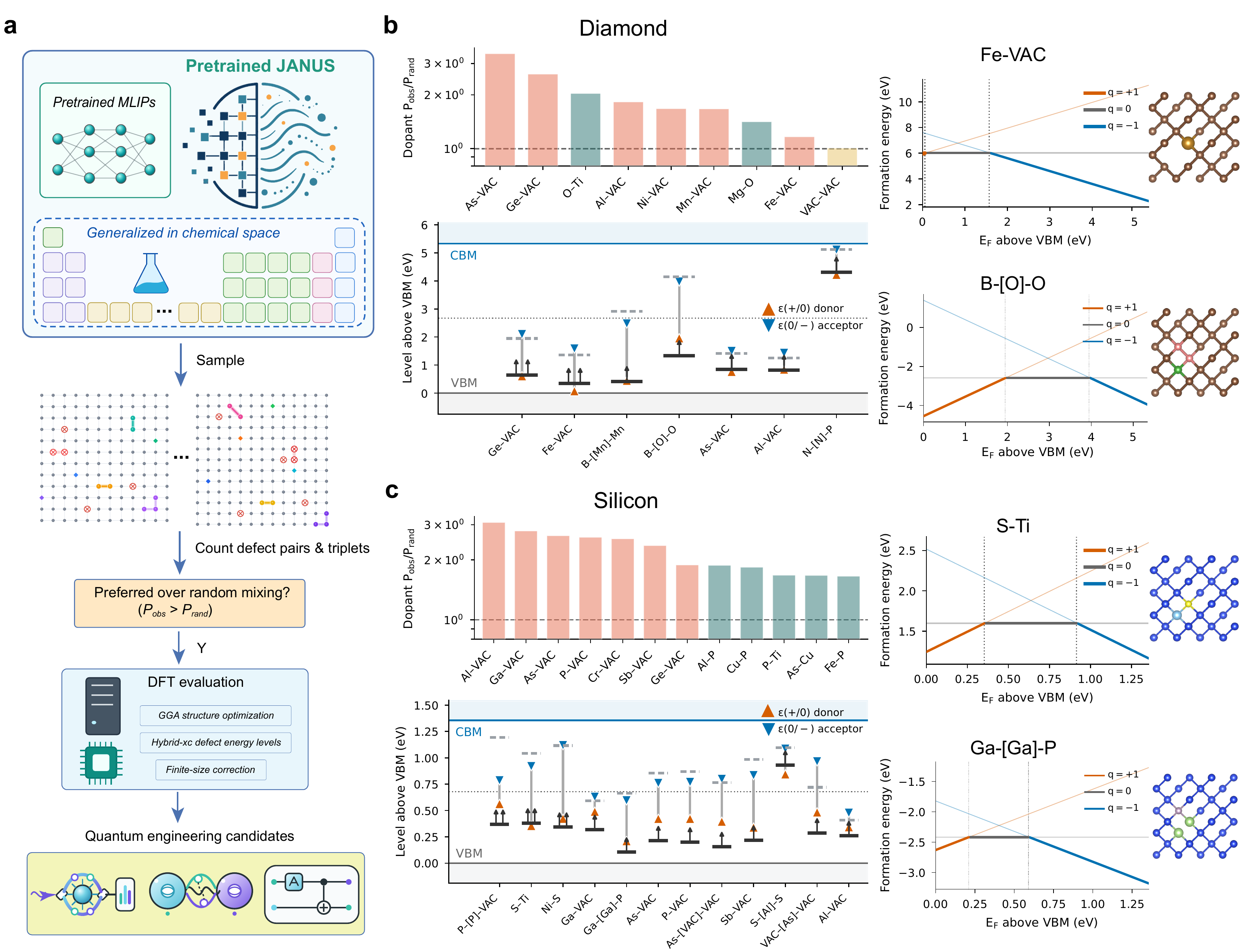}
  \caption{\textbf{Generative discovery of defect candidates for quantum engineering.} \textbf{a.} Schematic of the discovery workflow. An amortized JANUS sampler generates substitutional dopants, vacancies and multi-defect configurations across chemical space, under a pretrained machine learning interatomic potential.
  Nearest-neighbor pairs and connected triplets are extracted from the generated ensemble and compared with composition-matched random mixing. Enriched motifs are subsequently evaluated using density functional theory to identify potential quantum-engineering candidates.
  \textbf{b.} Discovery of quantum defects in diamond. Top, enrichment of selected defect motifs $P_\text{obs}/P_\text{rand}$. The dashed line denotes random mixing, and bar colors distinguish vacancy-dopant complexes, co-doped pairs and divacancy motifs. Bottom, calculated in-gap electronic states relative to the valence-band maximum (VBM) and conduction-band minimum (CBM), with the occupied HOMO shown by a solid black line and the unoccupied LUMO by a dashed grey line. The total spin of the neutral ground state is indicated alongside. Thermodynamic charge-transition levels (CTLs), $\varepsilon(+/0)$ and $\varepsilon(0/-)$, are denoted by orange and blue triangles, respectively. 
  Right, formation energies of the $q=+1,0,-1$ charge states w.r.t. Fermi level ($E_F$) for the highlighted Fe-VAC and B-[O]-O candidates, together with illustrations of their defect structures. Vertical dotted lines mark the CTLs. 
  \textbf{c.} Discovery of quantum defects in silicon, highlighting the S-Ti and Ga-[Ga]-P candidates. In the notation for defect triplets, the bracketed species denotes the central atom bonded to the two terminal defect atoms.}
  \label{fig4}
\end{figure}

\subsection*{Discovery of defect candidates for quantum engineering}
Point defects provide a rich platform for quantum sensing, communication and computation, but identifying useful quantum defects requires searching a combinatorial space of elemental substitutions, vacancies and multi-site complexes \cite{wolfowicz2021quantum,ping2021computational,thomas2024substitutional,fang2026towards}. 
We therefore use JANUS as a generative discovery engine to directly explore the chemical space of defects, as is shown in Fig.\,\ref{fig4}a. 
Starting from the pretrained machine-learning interatomic potential (MLIP) MACE-MPA-0\cite{batatia2022mace,batatia2025foundation}, a single amortized on-lattice grand-canonical neural sampler is trained to jointly generate chemical occupations and local atomic relaxations across a broad defect alphabet at temperature $T=800$ K.  Because the token of each lattice site is drawn conditionally on all previously revealed sites, JANUS can simultaneously sample isolated substitutions, vacancy complexes and higher-order multi-defect motifs, allowing favorable local chemical correlations to emerge directly from the target ensemble.
We apply this strategy to diamond and silicon, two prototypical semiconductors for quantum technology applications. Large ensembles generated by JANUS on supercells are decomposed into defect motifs, with a particular focus on pairs and connected triplets.
To distinguish genuine chemical association from random co-occurrence, we compared the observed frequency of each pair or triplet motif ($P_\text{obs}$), with its expected frequency under composition-matched random mixing ($P_\text{rand}$). The resulting enrichment ratio $P_\text{obs}/P_\text{rand}$ is used to screen and rank defect motifs for subsequent DFT validation. More details on the training and evaluation procedure are shown in Methods.

The top-ranked defect pairs, ranked by their enrichment relative to composition-matched random mixing, for diamond and silicon are shown in Fig.\,\ref{fig4}b,c. In diamond, JANUS reveals pronounced enrichment of several vacancy-dopant complexes (Fig.\,\ref{fig4}b). For example, the most enriched As-VAC occurs approximately threefold more frequently than expected from random mixing, while several additional vacancy complexes and heteroatomic co-doping motifs are also preferentially formed. 
Notably, this unbiased search recovers the Ge-VAC centre, an experimentally established colour centre in diamond with optically addressable spin states \cite{iwasaki2015germanium,siyushev2017optical}.
Similarly, in silicon, JANUS retrieves several canonical vacancy-dopant complexes, including the P-VAC $E$ centre and the experimentally identified As-VAC and Sb-VAC centres \cite{watkins1964defects,elkin1968defects}. 
Beyond rediscovery of known defect motifs, we further evaluate the enriched motifs in bulk silicon or diamond supercells using DFT, with structural relaxation followed by hybrid-functional single-point electronic structure calculations, including screening on charged states, spin states and finite-size corrections \cite{freysoldt2014first} (details in Methods). 
We aim to identify neutral defects with an open-shell neutral ground state ($S>0$), thermodynamic charge-transition levels (CTLs) defining a finite neutral charge-state stability window within the band gap, and an occupied in-gap defect state.
The top candidates satisfying these screening criteria are summarized for diamond and silicon in Fig.\,\ref{fig4}b,c. The resulting candidates span a remarkably broad chemical and structural space, including dopant-vacancy pairs, co-doped substitutional pairs, vacancy-containing triplets and purely substitutional three-site motifs, with elements drawn from groups III and V as well as transition metals. 

From these screened motifs, we highlight four representative candidates together with their charge-state formation energies in the rightmost column of Fig.\,\ref{fig4}b,c. In diamond, Fe-VAC ($S=1$) represents a transition-metal--vacancy centre. Its neutral ($q=0$) state is thermodynamically stable over a finite Fermi-level interval bounded by CTLs, while the associated occupied defect state lies inside the band gap, albeit closer to the VBM. 
By contrast, B-[O]-O ($S=1/2$) is a substitutional three-site complex and exhibits a substantially deeper neutral stability window, with both CTLs and the occupied deep in-gap state displaced further from the band edges. This more isolated electronic structure is particularly attractive for suppressing hybridization with bulk states and motivates further optical characterization.
Similarly in silicon, candidates such as S-Ti ($S=1$) and Ga-[Ga]-P ($S=1/2$) exhibit stable neutral charge states and in-gap defect energy levels despite arising from chemically and structurally distinct motifs. 
We note that our DFT calculations do not by themselves establish these motifs as definitive quantum-engineering candidates. Further validation would require calculating excited-state properties, such as zero-phonon-line energies for optical addressability, Huang-Rhys factors for electron-phonon coupling, and spin-relaxation for long-lived spin control.
Nevertheless, the JANUS discovery protocol provides a high-confidence first-principles shortlist at a fraction of the cost of exhaustive enumeration. By coupling grand-canonical generation, statistical motif discovery, and electronic structure validation, JANUS enables broad chemical-space discovery of unconventional clustered defect motifs.

\section*{Discussion}
In this work, we have introduced JANUS, a multimodal neural sampler that jointly treats discrete chemical identities and occupancies with continuous atomic displacements and cell-volume fluctuations, trained directly on the target interatomic potential without relying on precomputed reference data. This coupling enables direct sampling in on-lattice isobaric semi-grand-canonical $\Delta\mu NPT$ and grand-canonical $\mu VT$ ensembles, extending neural sampling beyond previous approaches restricted to either continuous $NVT$ ensembles or purely discrete configurations. 
To our knowledge, joint discrete-continuous thermodynamic sampling has not previously been realized by machine learning models. JANUS unifies these coupled degrees of freedom within a single amortized generative framework, enabling thermodynamic sampling, free energy and phase equilibrium calculations, conditional steering, inverse design and structural motif discovery across alloys and crystalline defects. Its sampling efficiency substantially reduces the number of energy evaluations required, making direct thermodynamic sampling with computationally demanding universal MLIPs practical.
A notable demonstration is the discovery of promising quantum-engineering defect candidates, including Fe-VAC and B-[O]-O in diamond and S-Ti and Ga-[Ga]-P in silicon. These capabilities in JANUS mark a crucial step towards efficient sampling, characterization and discovery of realistic disordered materials.

Nevertheless, two major challenges remain before universal thermodynamic neural sampling can be achieved. First, the present JANUS model is demonstrated on moderate lattice sizes, and the free energy and phase equilibrium demonstrations are primarily limited to binary alloys.
Second, JANUS is presently restricted to lattice-based solid phases with cubic simulation cells and therefore does not capture melting, liquid states or general non-cubic lattices. Extending the framework to flexible cell shapes and beyond crystalline basins would enable treatment of a broader range of crystal phases and complete solid-liquid phase diagrams, which are central to understanding experimental alloy synthesis, processing and phase selection \cite{kattner2016calphad}.
More generally, allowing particle insertion and deletion without predefined lattice sites would enable general off-lattice grand-canonical neural sampling, unlocking sampling scenarios such as gas adsorption, surface reactions and solid-liquid interfaces. 
Reaching this level of flexibility requires extending lattice-anchored displacements to fully off-lattice atomic coordinates, while jointly handling chemical identity, particle number, cell geometry and dynamically changing atomic coordination.
Ultimately, a universal neural sampler spanning these degrees of freedom and thermodynamic ensembles is the most desirable goal, which would require both methodological advances and larger-scale training.

Despite these limitations, JANUS provides a versatile foundation for substantially broader applications in disordered materials. Flexible thermodynamic sampling with JANUS, combined with conditional steering through inference-time reward tilting, could enable both structure determination and inverse design across diverse disordered materials.
For alloys and defects, where local atomic arrangements can be difficult to resolve experimentally \cite{he2024quantifying,cheng2026foundation}, experimental diffraction and spectroscopic signals could be incorporated as constraints on a generative prior to solve the inverse problem and infer compatible structural ensembles.
In parallel, rewards defined on structural, mechanical, electrical, magnetic, optical or chemical properties could steer the same pretrained sampler towards targeted functional materials without retraining for each objective.
Moreover, combined with the previous continuous-space sampler \cite{cheng2026atlas} for amorphous materials, JANUS suggests a path towards systems in which chemical and structural disorder coexist across phases and length scales, including oxide glasses, disordered ceramics, solid electrolytes and solid-liquid interfaces.
At even larger length scales, transferable local neural samplers could be deployed around extended defects such as grain boundaries, dislocations and heterogeneous interfaces, where chemical segregation, local reconstruction and strain are strongly coupled. 
By unifying multimodal discrete-continuous thermodynamic sampling, JANUS paves the way towards general neural sampling of disordered condensed matter.

\section*{Methods}
\subsection*{Stochastic interpolants on continuous and discrete domains}
As in the Overview, a configuration $(\mathbf a,\mathbf x)=(\mathbf a,\mathbf u,v)$ combines the discrete species channel $\mathbf a\in\mathcal A^N$, with $\mathcal A$ the token alphabet, and the continuous channels $\mathbf x=(\mathbf u,v)$---the fractional displacements $\mathbf u\in\mathbb R^{N\times3}$ and the log-volume $v\in\mathbb R$. In this subsection we fix a single thermodynamic state $(T,\Delta\mu)$, where $\Delta\mu$ collects the chemical-potential differences of the active tokens, and write $\pi_1$ for the target density of Eq.~\eqref{eq:target}; the amortized, conditional version is described below. JANUS constructs a generative process that transports a simple initial law $\pi_0$ at time $t=0$ onto $\pi_1$ at $t=1$, treating all three channels jointly: the continuous channels start from a product prior $\pi_{0,c}(\mathbf x)=\pi_{0,u}(\mathbf u)\,\pi_{0,v}(v)$, while the species channel starts from a \emph{fully masked} state, in which every site carries an auxiliary absorbing symbol $\texttt{M}\notin\mathcal A$ and acquires its species label only during generation. For the continuous channels, $\pi_{0,u}$ and $\pi_{0,v}$ are Gaussians: since we study crystal structures on a lattice below the melting point, atoms vibrate around their reference sites and the cell fluctuates narrowly around its equilibrium volume, so Gaussian laws centred on the reference lattice leave only a short prior-to-target transport for the network to learn; the physics-informed parameter choices used in practice are specified later. The construction extends stochastic interpolants \cite{albergo2025stochastic} to a mixed continuous--discrete state space.

\paragraph{Continuous channels.} For each continuous channel $c\in\{u,v\}$, writing $x^c$ for the corresponding component of $\mathbf x$, a prior draw $x_0^c$ and a terminal configuration $x_1^c$ are connected by the deterministic interpolant
\begin{equation}
x_t^c=(1-\alpha(t))\,x_0^c+\alpha(t)\,x_1^c,
\label{eq:cont-interpolant}
\end{equation}
where $\alpha(t)$ increases monotonically from $\alpha(0)=0$ to $\alpha(1)=1$, pinning $x_t^c$ to the prior draw at $t=0$ and to the terminal at $t=1$. The interpolant defines time-dependent marginals $p_t$ that connect $\pi_0$ and $\pi_1$; because the prior is Gaussian, $p_t$ has a smooth density at every intermediate time even though no noise is injected along the path. The stochastic dynamics realizing these marginals are built from two fields per channel: a velocity $b^c$ and a score $s^c=\nabla_{x^c}\log p_t$. Both are conditional expectations over the interpolant. The velocity is
\begin{equation}
b^c(\mathbf a,\mathbf x,t)=\mathbb E\big[\dot\alpha(t)\,(x_1^c-x_0^c)\,\big|\,(\mathbf a_t,\mathbf x_t)=(\mathbf a,\mathbf x)\big],
\label{eq:cont-velocity}
\end{equation}
where the conditioning is on the full state $(\mathbf a_t,\mathbf x_t)$, and the score satisfies the target score identity \cite{de2024target}
\begin{equation}
s^c(\mathbf a,\mathbf x,t)=\mathbb E\Big[\tfrac{1}{\alpha(t)}\,\nabla_{x_1^c}\log \pi_1(\mathbf a_1,\mathbf x_1)\,\Big|\,(\mathbf a_t,\mathbf x_t)=(\mathbf a,\mathbf x)\Big].
\label{eq:tsi}
\end{equation}
Hence, wherever an interpolant path visits, the target score evaluated at the path's endpoint provides an unbiased estimate of the marginal score there, so the physics enters the model directly through derivatives of the target density rather than through reference data.

\paragraph{Discrete channel.} The species channel is transported by an absorbing-mask interpolant \cite{austin2021structured, shi2024simplified}: given the terminal symbol $a_{1,i}$, each site is drawn independently from
\begin{equation}
\mathbb P\big(a_{t,i}=b'\,\big|\,a_{1,i}=b\big)=\alpha_a(t)\,\mathbf 1\{b'=b\}+\big(1-\alpha_a(t)\big)\,\mathbf 1\{b'=\texttt{M}\},
\label{eq:disc-interpolant}
\end{equation}
where the reveal schedule $\alpha_a(t)$ increases monotonically from $\alpha_a(0)=0$ to $\alpha_a(1)=1$ (we use $\alpha_a(t)=t$): at time $t$ a site shows its terminal symbol with probability $\alpha_a(t)$ and the mask otherwise, and no site ever shows a \emph{wrong} symbol. The endpoints are again pinned---the all-mask state at $t=0$, the terminal at $t=1$. The generative dynamics realizing these marginals is a pure reveal process: each site is unmasked exactly once, at a random time governed by $\alpha_a$, and never re-masked. A single learned object controls this process, the denoising posterior
\begin{equation}
q_{\theta,i}(b\mid \mathbf a,\mathbf x)\;\approx\;\mathbb P\big(a_{1,i}=b\,\big|\,(\mathbf a_t,\mathbf x_t)=(\mathbf a,\mathbf x)\big),
\label{eq:masked-head}
\end{equation}
the conditional law of site $i$'s terminal symbol given the current partially revealed state---one categorical network head per site. For a revealed site this posterior is a point mass at the shown symbol; for a masked site it is the model's completion law given everything revealed so far, and the exact posterior is independent of $t$ given the revealed set, since the mask carries no information beyond which sites are hidden \cite{ou2024your}. Whenever the schedule reveals a masked site during generation, its symbol is drawn from $q_{\theta,i}$ evaluated at the current state.

\paragraph{Numerical integration.} Samples are generated by integrating the coupled dynamics from a prior draw $(\mathbf a_0,\mathbf x_0)$, the all-mask state paired with $\mathbf x_0\sim\pi_0$, to $t=1$ on a shared grid $t_n=nh$ with step $h=1/M$. The continuous channels follow the generative stochastic differential equation \eqref{eq: fwd sde} discretized with the Euler--Maruyama scheme,
\begin{equation}
x_{n+1}^c=x_n^c+\big[b_\theta^c+g^2(t_n)\,s_\theta^c\big](\mathbf a_n,\mathbf x_n,t_n)\,h+\sqrt{2g^2(t_n)h}\;\xi_n^c,\qquad \xi_n^c\sim\mathcal N(0,I),
\label{eq:em}
\end{equation}
where $b_\theta^c$ and $s_\theta^c$ denote the network approximations of the velocity \eqref{eq:cont-velocity} and score \eqref{eq:tsi} (their training is described in the next subsection), and the diffusion strength $g(t)\ge0$ is a free parameter: the score-corrected drift leaves the marginals $p_t$ unchanged for every $g$, and $g=0$ recovers the deterministic probability-flow integrator. 

The species channel follows an absorbing-mask process with the linear reveal schedule \(\alpha_a(t)=t\). During \([t_n,t_{n+1}]\), every still-masked site is independently revealed with probability
\begin{equation}
\gamma_n=\frac{\alpha_a(t_{n+1})-\alpha_a(t_n)}{1-\alpha_a(t_n)}.
\label{eq:reveal}
\end{equation}
If site $i$ is revealed, its symbol is drawn from
$a_{n+1,i}\sim q_{\theta,i}\left(\cdot\mid\mathbf a_n,\mathbf x_n,t_n\right)$,
and revealed sites remain fixed thereafter. This is a tau-leap discretization: throughout one interval, all reveals use the posterior evaluated at the interval’s left-end state and are conditionally independent. Consequently, several sites may be revealed simultaneously. As \(h\to0\), simultaneous reveals vanish in probability and the construction approaches sequential autoregressive unmasking, with each reveal conditioned on the previously revealed symbols and the contemporaneous continuous state. Both channel updates use the same network evaluation at \((\mathbf a_n,\mathbf x_n,t_n)\); the evaluation at the resulting state is then reused for the next step.

\subsection*{Data-free fixed-point training}
All learned fields of JANUS are the conditional expectations introduced in the previous subsection: the velocity \eqref{eq:cont-velocity} and score \eqref{eq:tsi} are conditional expectations over the interpolant, and the denoising posterior \eqref{eq:masked-head} is the conditional law of a site's terminal symbol. Estimating any of them requires terminal configurations $(\mathbf a_1,\mathbf x_1)\sim\pi_1$---precisely what a sampler is meant to produce and what is unavailable a priori. JANUS resolves this, without using any reference configurations, by a data-free fixed-point iteration in the spirit of bridge-matching samplers \cite{blessing2026bridge} and related fixed-point diffusion samplers \cite{havens2025adjoint,havens2026flow,liu2026adjoint}: the unavailable terminals in the training targets are replaced by the model's own samples, and the physics enters solely through the labels attached to these self-generated configurations---forces, cell-rescaling derivatives, and single-site heat-bath conditionals, all pointwise evaluations of the interatomic potential. We first specify the regression objectives evaluated on such terminals, then the iteration itself.

\paragraph{Regression objective.} Because the continuous fields \eqref{eq:cont-velocity} and \eqref{eq:tsi} are conditional expectations, they can be learned by least-squares regression onto unbiased single-sample estimates. With a prior draw $(\mathbf x_0,t)$ and a terminal $(\mathbf a_1,\mathbf x_1)$ sampled, and the intermediate state $(\mathbf a_t,\mathbf x_t)$ formed via \eqref{eq:cont-interpolant} and \eqref{eq:disc-interpolant}, the per-sample targets are
\begin{equation}
\hat s^{\,c}=\frac{1}{\alpha(t)}\,\nabla_{x_1^c}\log\pi_1(\mathbf a_1,\mathbf x_1),\qquad
\hat b^{\,c}=\dot\alpha(t)\,(x_1^c-x_0^c),
\label{eq:cont-targets}
\end{equation}
and the continuous loss is
\begin{equation}
\mathcal L_{\mathrm{cont}}(\theta)=\sum_{c\in\{u,v\}}\mathbb E\Big[\big\|b_\theta^c(\mathbf a_t,\mathbf x_t,t)-\hat b^{\,c}\big\|^2+\big\|s_\theta^c(\mathbf a_t,\mathbf x_t,t)-\hat s^{\,c}\big\|^2\Big],
\label{eq:loss-cont}
\end{equation}
whose minimizers are exactly \eqref{eq:cont-velocity} and \eqref{eq:tsi}. In practice, the variance of $\hat s^{\,c}$ can be further reduced; see Supplementary Section~\ref{si:gtsi}.

\paragraph{Soft cross-entropy objective.} The species head is trained via an analogous fixed-point strategy, but based on a cross-entropy objective. This masked cross-entropy objective mirrors the one from masked discrete diffusion models
\cite{shi2024simplified}, but instead of relying on one-hot labels encoding the ground-truth terminal symbol, the target marginal density (soft label) for each site is the conditional Boltzmann distribution of this site given the terminal symbols at all other sites. Namely, if we construct $\mathbf a_t$ by masking each symbol in $\mathbf a_1$ with probability $1-\alpha_a(t)$ independently, the soft cross-entropy loss reads
\begin{equation}
\mathcal L_{\mathrm{disc}}(\theta)=-\,\mathbb E\Bigg[\sum_{i:\,a_{t,i}=\texttt{M}}\ \sum_{b\in\mathcal A}\ell_i(b)\,\log q_{\theta,i}(b\mid \mathbf a_t,\mathbf x_t)\Bigg],\qquad
\ell_i(b)=\rho_i(b\mid\mathbf a_{1,-i},\mathbf x_1),
\label{eq:loss-disc}
\end{equation}
where $\rho_i(\,\cdot\mid\mathbf a_{1,-i}, \mathbf x_1)$ is the target's single-site (heat-bath) conditional at the terminal,
\begin{equation}
\rho_i\big(b\,\big|\,\mathbf a_{1,-i},\mathbf x_1 \big)
=\frac{\exp\!\big[\big(\mu_b-U(\mathbf a_1^{(i\to b)},\mathbf x_1)\big)/k_{\mathrm B}T\big]}
{\sum_{b'\in\mathcal A}\exp\!\big[\big(\mu_{b'}-U(\mathbf a_1^{(i\to b')},\mathbf x_1)\big)/k_{\mathrm B}T\big]},
\label{eq:flip-score}
\end{equation}
where $\mathbf a_1^{(i\to b)}$ denotes the terminal species configuration with site $i$ set to $b$; every factor of the target \eqref{eq:target} that does not depend on $a_{1,i}$ cancels between numerator and denominator, and evaluating the $K$ single-site substitution energies requires one batched call to the potential. For a binary alphabet, \eqref{eq:flip-score} reduces to a sigmoid of the flip log-ratio $(\Delta\mu\,\Delta N_{\mathrm B}-\Delta U_i)/k_{\mathrm B}T$, where $\Delta U_i$ is the change in potential energy under the flip at site $i$. The label $\ell_i$ does not depend on the terminal symbol $a_{1,i}$ actually stored---the conditional is a function of the other sites only---and this invariance is precisely what makes it a genuine conditional: at the Boltzmann fixed point, averaging $\ell_i$ over the posterior of the remaining terminal sites given $(\mathbf a_t,\mathbf x_t)$ reproduces the exact denoising posterior, so the minimizer of \eqref{eq:loss-disc} is \eqref{eq:masked-head}. Notably, the energy enters only through the bounded factor $\rho_i\in(0,1)$, rather than appearing directly in the loss.

\paragraph{Fixed-point iteration.} The sampler is additionally amortized over thermodynamic conditions\cite{schebek2024efficient,cheng2026atlas}: all fields are heads of one site-equivariant message-passing network that receives $(T,\Delta\mu)$ as additional inputs (Supplementary Information), so a single model represents the whole family of ensembles $\pi_1(T,\Delta\mu)$ over a prescribed window rather than a single state point. Each training round consists of four steps (Fig.~\ref{fig:m-fixedpoint}). (i)~\emph{Generate}: draw per-chain conditions $(T,\Delta\mu)\sim p(T,\Delta\mu)$---the construction of the conditioning distribution $p$ is described in Supplementary Section~\ref{si:conditioning}---and integrate \eqref{eq:em}--\eqref{eq:reveal} under these conditions to produce terminal configurations. (ii)~\emph{Label}: evaluate at each terminal the quantities entering the targets---the continuous scores \eqref{eq:cont-targets} and the single-site heat-bath conditionals \eqref{eq:flip-score}---and store them, together with the conditions, in a replay buffer. (iii)~\emph{Interpolate}: draw fresh $(\mathbf x_0,t)$ for the continuous channels and form intermediate states $(\mathbf a_t,\mathbf x_t)$ toward buffered terminals via \eqref{eq:cont-interpolant}, masking each buffered species site independently with probability $1-\alpha_a(t)$ per \eqref{eq:disc-interpolant}. (iv)~\emph{Regress}: minimize the conditional objective
\begin{equation}
\mathcal L(\theta)=\mathbb E_{(T,\Delta\mu)\sim p}\Big[\mathcal L_{\mathrm{cont}}(\theta;T,\Delta\mu)+\lambda\,\mathcal L_{\mathrm{disc}}(\theta;T,\Delta\mu)\Big],
\label{eq:cond-loss}
\end{equation}
where both losses are those introduced above, with all fields and targets evaluated at the stored per-sample conditions, and $\lambda>0$ balances the channels. Under idealized population training, the family of Boltzmann ensembles is a fixed point of this iteration: if the model samples every $\pi_1(T,\Delta\mu)$ exactly and the learned heads $b_\theta^c$, $s_\theta^c$ and $q_{\theta}$ match the conditional expectations of the regression and cross-entropy targets, training is stationary.
For the continuous channels this self-consistency result is inherited from fixed-point diffusion samplers \cite{blessing2026bridge}; for the discrete channel, the corresponding iteration---fitting the denoising posterior by cross-entropy against heat-bath labels evaluated at self-generated terminals--- is, to our knowledge, new; its convergence properties are analyzed in a companion work \cite{he2026fixed}. In the Supplementary Information we formalize the discrete iteration and show that the Boltzmann ensemble is indeed a fixed point of it, with the heat-bath conditionals of the target acting as the anchor (Supplementary Section~\ref{si:discrete-fixed-point}).
Whether the Boltzmann ensemble is the unique fixed point of the iteration is a stronger question that we defer to future work.
Intuitively, the force and heat-bath labels should be the guiding mechanisms toward this fixed point, since each self-generated configuration, however far from equilibrium, is labeled with the true local direction of increasing target probability.

\paragraph{Physics-informed priors.} The masked species channel starts from the all-mask state, so the prior is a design choice for the continuous channels only. 
Both continuous channels use Gaussian priors informed by the thermodynamic conditions. For the log-volume channel, the prior mean is set by an estimate of the per-atom volume at temperature $T$ and composition $c_0$, where $c_0$ is obtained from an inexpensive linear approximation to the phase-transition line in the $(T,\Delta\mu)$ plane. For the displacement channel, the prior variance scales linearly with temperature, consistent with harmonic fluctuations, $\langle u^2\rangle\propto T$.
All prior parameters are derived from the interatomic potential alone, so no reference data enters training; the full construction and the parameter fits are given in Supplementary Section~\ref{si:prior-fit}.

\subsection*{Semi-grand partition function estimation and observable reweighting}

The trained sampler gives direct access to the semi-grand partition function $\Xi(T,\Delta\mu,P)$---the normalization of the target ensemble \eqref{eq:target}---by importance sampling \cite{he2026free}; the semi-grand potential follows as $\Phi=-k_{\mathrm B}T\log\Xi$. 
The partition function is estimated using importance weights on the joint discrete–continuous trajectory space. Both channels have transition probabilities that can be evaluated explicitly. For the species channel, each site contributes one categorical probability when its terminal identity is revealed---the masked sampler is an autoregressive model in a random order \cite{ou2024your}. Along a generated trajectory $\omega=(\mathbf a_{0:M},\mathbf x_{0:M})$, their product is
\begin{equation}
\log q_{\mathrm{disc}}=\sum_{i=1}^N\log q_{\theta,i}\big(a_{1,i}\,\big|\,\mathbf a_{n_i},\mathbf x_{n_i}\big),
\label{eq:masked-qs}
\end{equation}
where $n_i$ is the step at which site $i$ was revealed.
This quantity collects the realized categorical transition factors; it is not, in general, the marginal likelihood of the terminal species configuration. A schedule-defined backward re-masking kernel cancels the reveal-schedule factors over the complete trajectory, leaving $-\log q_{\mathrm{disc}}$ as the discrete contribution to the forward–backward log-weight.
For the continuous channels, the weight contains ratios of backward and forward Gaussian transition densities.
The generative SDE \eqref{eq: fwd sde} admits a time reversal\cite{anderson1982reverse} built from the same learned fields, in which the drift subtracts rather than adds the score term,
\begin{equation}
\mathrm{d}x^c_t=\big[b^c_\theta-g^2(t)\,s^c_\theta\big](\mathbf a_t,\mathbf x_t,t)\,\mathrm{d}t+\sqrt{2}\,g(t)\,\mathrm{d}\bar w^c_t,
\label{eq:bwd-sde}
\end{equation}
integrated backward from $t=1$ to $t=0$ with a reversed-time Wiener process $\bar w^c_t$: forward and backward dynamics share the same marginals and differ only in the sign of the score term. Discretizing either direction with the Euler--Maruyama scheme of \eqref{eq:em} renders every one-step transition Gaussian, i.e., $ N\big(x_{n+1}^c;\ x_n^c+\big[b_\theta^c+g^2(t_n)\,s_\theta^c\big](\mathbf a_n,\mathbf x_n,t_n)\,h,\ 2g^2(t_n)h\big)$, so the forward--backward path weight of the continuous channels is a product of closed-form Gaussian density ratios. Each generated trajectory $(\mathbf a_{0:M},\mathbf x_{0:M})$ then carries the log weight
\begin{equation}
\log W=\log\pi_1(\mathbf a_M,\mathbf x_M)-\log\pi_0(\mathbf x_0)-\log q_{\mathrm{disc}}+\sum_{n=0}^{M-1}\Delta_n^{uv},
\label{eq:fbrnd}
\end{equation}
 where $\log\pi_1$ is evaluated in the sampled coordinates and includes the log-Jacobian contribution $Nv$ arising from the transformation from fractional to Cartesian positions
(Supplementary Section~\ref{si:volume-jacobian}),
the conditioned prior density $\pi_0$ is available in closed form (the species channel starts from a point mass at the all-mask state and contributes no prior term), and $\Delta_n^{uv}$ is the forward--backward path weight of the step $[t_n,t_{n+1}]$: the log-ratio of the backward and forward Gaussian transition densities, with the backward drift of \eqref{eq:bwd-sde} evaluated at the arrival state,
\begin{equation}
\Delta_n^{uv}=\sum_{c\in\{u,v\}}\log\frac{\mathcal N\big(x_n^c;\ x_{n+1}^c-\big[b_\theta^c-g^2(t_{n+1})\,s_\theta^c\big](\mathbf a_{n+1},\mathbf x_{n+1},t_{n+1})\,h,\ 2g^2(t_n)h\big)}{\mathcal N\big(x_{n+1}^c;\ x_n^c+\big[b_\theta^c+g^2(t_n)\,s_\theta^c\big](\mathbf a_n,\mathbf x_n,t_n)\,h,\ 2g^2(t_n)h\big)}.
\label{eq:fbrnd-uw}
\end{equation}
 The realized categorical reveal probabilities of $\mathbf a$ and the Gaussian path-space ratio for $(\mathbf u,v)$ are accumulated along a single trajectory and multiply into one importance weight---the free energy estimator is hybrid in the same sense as the sampler itself. The split matters quantitatively: the species term \eqref{eq:masked-qs} is a sum of $N$ realized log-probabilities whose fluctuations reflect the model's error in configuration space, not a sum of $M\times N$ per-step flip log-ratios, so the discrete part of the weight does not degrade as the integration grid is refined. Over $n_{\mathrm{traj}}$ independent trajectories, the weights give an asymptotically unbiased estimate of the partition function \cite{jarzynski1997nonequilibrium},
\begin{equation}
\Xi=\mathbb E\big[e^{\log W}\big],\qquad
\log\widehat\Xi=\operatorname*{logsumexp}_{j\le n_{\mathrm{traj}}}\ \log W^{(j)}-\log n_{\mathrm{traj}}.
\label{eq:fbrnd-Z}
\end{equation}
Because one amortized model covers the whole $(T,\Delta\mu)$ window, Eq.~\eqref{eq:fbrnd-Z} directly yields partition-function ratios---and hence semi-grand-potential differences---across thermodynamic conditions from a single training run. The same weights also serve for reweighting: for any observable $O$ of the terminal configuration, the self-normalized estimator
\begin{equation}
\widehat{\langle O\rangle}_{\pi_1}=\frac{\sum_{j}e^{\log W^{(j)}}\,O\big(\mathbf a_M^{(j)},\mathbf x_M^{(j)}\big)}{\sum_{j}e^{\log W^{(j)}}}
\label{eq:reweighted-obs}
\end{equation}
converges to the exact ensemble average $\langle O\rangle_{\pi_1}$, so it remains asymptotically unbiased under model imperfection---residual model error is traded for weight variance rather than bias.

\subsection*{Mixing free energy and binodal estimation}
Phase equilibria follow from resolving the semi-grand ensemble by composition. Writing $n=N_{\mathrm B}$ and $x=n/N$, the semi-grand partition function $\Xi(T,\Delta\mu,P)$ decomposes into canonical (fixed-composition) isothermal--isobaric partition functions $Z_n$,
\begin{equation}
\Xi(T,\Delta\mu,P)=\sum_{n=0}^{N}e^{\Delta\mu\,n/k_{\mathrm B}T}\,Z_n(T,P),\qquad
Z_n(T,P)=\sum_{\mathbf a:\,N_{\mathrm B}(\mathbf a)=n}\ \int\!\mathrm d\mathbf u\,\mathrm dv\ e^{-\left[U(\mathbf a,\mathbf u,v)+PV\right]/k_{\mathrm B}T},
\label{eq:xi-decomp}
\end{equation}
and the composition-resolved Gibbs free energy is $G(n)=-k_{\mathrm B}T\log Z_n$, defined up to an $n$-independent constant. The free energy of mixing per atom is obtained by subtracting the chord between the pure-phase free energies,
\begin{equation}
G_{\mathrm{mix}}(x)=\frac1N\Big[G(xN)-(1-x)\,G(0)-x\,G(N)\Big],
\label{eq:fmix}
\end{equation}
which cancels every term linear in $n$, so $G_{\mathrm{mix}}$ is invariant both under that constant and under any $\Delta\mu$ tilt of the ensemble in which $G$ was measured.

At each temperature, phase coexistence is determined from the lower convex envelope of the relevant composition-resolved free energy curves. Where a common tangent connects two distinct compositions, its contact points identify the coexisting compositions and its slope gives the coexistence chemical-potential difference $\Delta\mu^{*}(T)$. Repeating this construction across temperatures traces the binodal (coexistence) curve.
We reconstruct $G(n)$ using two complementary approaches. For Cu--Ni, we use an inverse Legendre--Fenchel construction, maximizing $\Phi(T,\Delta\mu)+\Delta\mu\,n$ over a grid of chemical potentials at each composition. This provides an approximate finite-cell reconstruction of the convex free energy. To resolve free energy profiles within miscibility gaps, we instead use fixed-composition sampling for Cu--Ag and Ni--Cr, reconstructing $G(n)$ from free energy differences between neighboring compositions estimated using single-site substitution energies. For Ni--Cr, pure-element absolute free energy obtained by Frenkel--Ladd thermodynamic integration \cite{frenkel1984new} place the fcc and bcc branches on a common reference scale. Details of both reconstruction approaches, the absolute-reference calculations and the Monte Carlo reference data are provided in Supplementary Section~\ref{si:free-energy}.

\subsection*{Inference-time steering to tilted ensembles}
It is often desirable to sample not the learned ensemble itself but a tilted version of it,
\begin{equation}
p_\eta(\mathbf a,\mathbf x)\;\propto\;p_{\text{base}}(\mathbf a,\mathbf x)\,e^{\eta\,r(\mathbf a,\mathbf x)},
\label{eq:tilt}
\end{equation}
where $p_\text{base}$ is the terminal distribution of the trained sampler, $r$ is a reward defined on terminal configurations and $\eta$ controls the strength of the tilt.
When many different tilts of the same base model are needed, one can tilt at inference time, without retraining. 
We generate configurations using sequential Monte Carlo: particles evolve under proposal dynamics and are reweighted at each generation step, with resampling as needed, so that the final weighted ensemble approximates $p_\eta$.

To introduce the bias during generation, we define an intermediate
reward $r_t$ satisfying $r_0\equiv0$ and $r_1=\eta r$. These endpoint
conditions leave the prior unchanged and set the terminal target
to $p_\eta$, while the intermediate values determine how the bias
is introduced along the trajectory. We use two constructions for $r_t$. The first evaluates the reward
on the current configuration, $r_t(\mathbf a,\mathbf x)
=\eta\,\alpha_r(t)\,r(\mathbf a,\mathbf x)$,
where $\alpha_r(0)=0$ and $\alpha_r(1)=1$. This construction is used
when the reward can be evaluated on partially generated states,
with masked species represented by the probabilities predicted
by the species head. The second evaluates the reward on an estimate
of the terminal configuration. For each continuous channel, the
velocity head provides a posterior-mean (Tweedie)
estimate \cite{robbins1992empirical},
$\hat x^c_1
=x^c+\frac{1-\alpha(t)}{\dot\alpha(t)}
\,b^c_\theta(\mathbf a,\mathbf x,t)$.
Revealed species retain their assigned identities, while masked species are represented by the predicted terminal probabilities $q_{\theta,i}$ of Eq.~\eqref{eq:masked-head}. The reward is evaluated using these estimated terminal coordinates and species probabilities, and multiplied by the same factor $\eta\,\alpha_r(t)$.

Particles evolve under proposal dynamics that may differ from those of
the pretrained sampler. We retain the pretrained species update,
including both the reveal probabilities and the distributions of revealed
symbols. For each continuous channel, we retain the pretrained noise
amplitude but allow the drift to change. To favor higher intermediate
rewards, we use
\begin{equation}
a^c=b^c_\theta+g^2(t)s^c_\theta
+\lambda_c(t)\nabla_{x^c}r_t,
\label{eq:guided-drift}
\end{equation}
where $\lambda_c(t)$ controls the guidance strength. Setting
$\lambda_c=0$ recovers the pretrained propagation dynamics and requires
no reward gradients. After each propagation step, we update the particle weights to introduce
the reward bias and compensate for any change in the dynamics. Writing
$y_n=(\mathbf a_n,\mathbf x_n)$, and denoting the pretrained and proposal Gaussian transition kernels by $\overrightarrow{p}_{\theta,n}$ and $\overrightarrow{q}_{\mathrm{prop},n}$,
respectively, the update is
\begin{equation}
\Delta\log w_n
=r_{t_{n+1}}(y_{n+1})-r_{t_n}(y_n)
+\log\frac{\overrightarrow{p}_{\theta,n}(y_{n+1}\mid y_n)}
{\overrightarrow{q}_{\mathrm{prop},n}(y_{n+1}\mid y_n)},
\label{eq:smc-weight}
\end{equation}
where $\overrightarrow{q}_{\mathrm{prop},n}$ is a Gaussian transition as in \eqref{eq:fbrnd-uw} with an arbitrary proposal drift.
The first term introduces the change in reward, while the transition
ratio corrects for sampling from the proposal rather than the pretrained dynamics. Because the species updates are identical, their factors cancel from this ratio, leaving only the explicitly evaluable Gaussian transition densities of the continuous channels.

Particles are resampled when the effective sample size falls below a
threshold. Weighting and resampling therefore steer both species and
continuous variables, even though gradient guidance modifies only the
continuous dynamics. Without guidance, the transition ratio equals one,
and steering operates through reward-based weighting and resampling
alone. With the corresponding weight correction, the intermediate reward
and guidance strength can be chosen to improve sampling efficiency
without changing the terminal target; the finite-particle ensemble
remains an approximation to that target. The weight derivation and
explicit Gaussian ratios are given in Supplementary
Section~\ref{si:steering}, and guidance of the species proposal is
described in Supplementary Section~\ref{si:steering-general}.

\subsection*{Model architecture and size-transfer}
The velocities and scores of the two continuous channels and the species posterior are heads of a single equivariant message-passing network with a shared PaiNN backbone \cite{schutt2021equivariant} adapted to the three channels. Nodes are lattice sites carrying invariant scalar and equivariant vector features; messages are exchanged on a graph rebuilt from the live configuration at every evaluation, with pairwise minimum-image displacements measured in physical units at the current box length $L=e^{v/3}$, so the connectivity follows the fluctuating cell. Pair features expand the interatomic distance in a Gaussian radial basis under a smooth polynomial cutoff envelope, with the cutoff radius matched to the interaction range of the potential ($5.0$--$5.3$\,\AA{} for the alloys studied here).

Because every learnable parameter is size-intrinsic---per-site embeddings, per-edge messages confined to a fixed physical cutoff, per-site heads, and intensive pooled read-outs---the same weights define a sampler for any supercell: a model trained on a cell of $N$ lattice sites is instantiated on a cell of $N'$ sites by replacing only the reference lattice. The channels then transform according to their physical scaling: per-site quantities (species logits, displacement fields in physical units) carry over unchanged, while the collective log-volume channel is emitted as an intensive field and converted with the instantiated cell's $\sqrt{N'}$---its equilibrium fluctuations shrink as $1/\sqrt{N'}$ and its score grows as $\sqrt{N'}$---with the prior widths of Eqs.~\eqref{eq:prior-w}--\eqref{eq:prior-u} rescaled by the same factors. Under these conventions the network is by construction exactly equivariant under periodic tiling of the simulation cell---per-site outputs are unchanged and the volume outputs rescale by $\sqrt{N/N'}$---provided the cutoff respects the minimum-image bound $r_{\mathrm{cut}}\le L/2$. A model trained on a moderate cell can therefore be evaluated directly on substantially larger cells without retraining\cite{schebek2026scalable}.

\subsection*{Target systems, interatomic potentials and property evaluation}
The binary Cu--Ni and Cu--Ag alloys are described using embedded-atom method (EAM) potentials \cite{fischer2019systematic,williams2006embedded}, and the binary Ni--Cr and ternary Cr--Co--Ni alloys using a Finnis--Sinclair-type Ni--Co--Cr potential \cite{li2019strengthening}, all available in the NIST Interatomic Potentials Repository \cite{hale2018evaluating}.
Alloys containing up to three elements selected from the Ag--Al--Au--Cu--Ni chemical space, together with all systems with defects, are described using the universal MACE-MPA-0 MLIP.
Unless otherwise stated, fcc alloys are represented by $3\times3\times3$ conventional supercells containing 108 sites, bcc alloys by $4\times4\times4$ conventional supercells containing 128 sites, and diamond-cubic defect systems by $3\times3\times3$ supercells containing 216 sites.
All isobaric SGC calculations are performed at zero external pressure, $P=0$, appropriate for ambient-pressure solids where atmospheric pressure is negligible relative to characteristic internal elastic stresses.

Regarding property evaluation, chemical SRO is quantified using the Warren-Cowley parameter $\alpha_{ij}=1-P_{ij}/c_j$, where $P_{ij}$ is the probability that a nearest neighbour of species $i$ is species $j$ and $c_j$ is its global concentration. Nearest-neighbour shells are defined using the first minimum of the corresponding radial distribution function. 
Mechanical properties are evaluated following the procedure used in ATLAS \cite{cheng2026atlas}. Generated structures are first fully relaxed using MACE-MPA-0, after which the bulk modulus $B$ is obtained from an isotropic Birch-Murnaghan equation-of-state fit and the shear modulus $G$ from orthorhombic and monoclinic volume-conserving deformations combined using the Voigt-Reuss-Hill average.
For optical-property prediction, a reference dataset of frequency-dependent dielectric responses for 2,048 alloy structures is first generated using DFT calculations. GNNOpt \cite{hung2024universal}, an equivariant GNN using ensemble atomic embeddings, is then trained to predict the optical spectrum from atomic structures, from which the frequency-dependent reflectance and visible-range-averaged reflectance $R_\text{visible}$ are obtained. 
For point defects, JANUS is used only for structural-motif screening, and the resulting candidates are subsequently evaluated individually using DFT for defect formation energies, CTLs and in-gap electronic states as described below. The full details of DFT calculations are shown in ``First-principles DFT calculations'' below.

Multi-objective inverse design follows the LLM-driven evolutionary algorithm (LLM-EA) approach introduced in ATLAS \cite{cheng2026atlas,wang2025efficient}. The LLM proposes chemical-potential differences $\Delta\mu$ between chemical elements, for which JANUS generates an ensemble of 512 configurations to evaluate the corresponding equilibrium composition and target properties as ensemble averages.
Evaluations are managed asynchronously across up to 16 concurrent jobs, with new proposals launched as computational slots become available. The resulting $\Delta\mu$, composition and property tuples are accumulated in the search history.
For each proposal, 10 previous evaluations are sampled as LLM context: 3 high-performing candidates ranked using randomly weighted normalized objectives, 4 sampled from the current Pareto set, and 3 randomly sampled dominated candidates for exploration. These examples are then used to condition subsequent chemical-potential proposals.

\subsection*{Grand-canonical sampling and identification of defect motifs}

Defects in diamond and silicon are sampled using JANUS in the fixed-volume grand-canonical ($\mu VT$) ensemble. We use 216-site $3\times3\times3$ diamond-cubic supercells at $T=800$ K, with the simulation-cell volume fixed at the pristine-host value. Each lattice site carries a discrete token corresponding to the host atom, one of 15 dopants, or a vacancy, together with a continuous displacement $\mathbf{u}_i$ from its ideal lattice position. For silicon, the 17-token alphabet contains Si, B, Al, Ga, C, Ge, P, As, Sb, S, Ti, Cr, Fe, Co, Ni, Cu and VAC; for diamond, it contains C, B, Al, N, P, As, O, S, Si, Ge, Mg, Ti, Cr, Mn, Fe, Ni and VAC.

The unnormalized target distribution is
\begin{equation}
\log \pi_1(\mathbf{a},\mathbf{u})=\beta\left[\sum_i \mu_{a_i}-U(\mathbf{a},\mathbf{u})-V_{\mathrm{cap}}(\mathbf{a})\right]+\sum_{i:a_i=\mathrm{VAC}}\log \rho_g(\mathbf{u}_i),
\end{equation}
where $\rho_g$ is a normalized site-local Gaussian distribution assigned to the auxiliary coordinates of vacancy sites. Vacancy sites are removed from the MACE interaction graph, such that these auxiliary coordinates do not contribute to the physical energy. Although the number of lattice sites is fixed, the number of physical atoms fluctuates through the vacancy population, making this representation equivalent to a lattice grand-canonical ensemble. The corresponding grand-canonical measure and treatment of the vacancy chemical potential are described in Supplementary Information.

The energy $U$ is evaluated using MACE-MPA-0 with approximately 9M parameters \cite{batatia2022mace,batatia2025foundation}. We fix the chemical-potential gauge by setting $\mu_{\mathrm{host}}=0$ and initialize the chemical potential of each defect species $X$ as
\begin{equation}
\mu_X=\Delta E_X+k_{\mathrm B}T\log\left(\frac{c_X}{c_{\mathrm{host}}}\right),
\end{equation}
where $\Delta E_X$ is the unrelaxed energy change upon replacing one host site by species $X$. We use $c_X=0.01$ and $c_{\mathrm{host}}=0.99$. This calibration places different defect species within accessible chemical-potential ranges and is not intended to reproduce experimental equilibrium concentrations.
To restrict sampling to the dilute-defect regime, we introduce the concentration penalty
\begin{equation}
V_{\mathrm{cap}}=k_{\mathrm{cap}}\,\mathrm{softplus}\left(N_{\mathrm{VAC}}-N_{\mathrm{VAC}}^{\max}\right)+k_{\mathrm{cap}}\sum_{X\in\mathcal{A}_{\mathrm{dopant}}}\mathrm{softplus}\left(N_X-N_X^{\max}\right),
\end{equation}
where $k_{\mathrm{cap}}=8$ eV, $\mathrm{softplus}(x)=\frac{1}{s}\log\left(1+e^{sx}\right)$, with $s=4$, $N_{\mathrm{VAC}}^{\max}=4$, and $N_X^{\max}=6$. 

To amortize sampling across chemical space, each configuration is conditioned on the host together with a randomly selected subset of up to three active defect species, including vacancy. For each active species, the chemical potential is sampled uniformly from $[\mu_X-0.20\,\mathrm{eV},\,\mu_X+0.20\,\mathrm{eV}]$, whereas inactive species are assigned $\mu=-12$ eV. Active-set sizes, including the host, are sampled from 2, 3 and 4 with probabilities 0.2, 0.4 and 0.4, respectively.

After training, 8,000 configurations are generated for each host and mapped onto nearest-neighbour defect graphs. We extract nearest-neighbour pairs and connected three-site motifs; for a triplet A-[B]-C, the bracketed species denotes the central site bonded to both terminal defects. For each motif, we compute its occurrence probability $P_{\mathrm{obs}}$ in the generated ensemble and compare it with the probability $P_{\mathrm{rand}}$ expected from random mixing at the sampled compositions. The enrichment ratio $P_{\mathrm{obs}}/P_{\mathrm{rand}}$ therefore quantifies preferential defect association, with values above unity indicating enrichment relative to random mixing. Enriched motifs are subsequently evaluated using DFT to identify candidates for quantum-engineering applications.

\subsection*{First-principles DFT calculations}
First-principles DFT calculations are performed using GPAW \cite{mortensen2024gpaw} with a plane-wave cutoff of 500 eV. For optical calculations, 2,048 fcc binary and ternary alloy structures containing up to 64 atoms are generated as special quasirandom structures (SQSs) \cite{zunger1990special} from the  Ag--Al--Au--Cu--Ni chemical space. An 800 K configuration is generated for each structure by NVT MD simulation and subsequently evaluated using the Perdew-Burke-Ernzerhof (PBE) functional \cite{perdew1996generalized} with a $3\times 3\times 3$ k-point mesh. 
The frequency-dependent dielectric response is calculated within the random-phase approximation (RPA) \cite{yan2011linear}. Optical spectra are evaluated over photon energies of 0-6 eV using a spectral broadening of 0.10 eV, and the visible-range reflectance is averaged over 1.65-3.10 eV.

For evaluation of defects for quantum engineering, selected defect motifs are evaluated in 216-site $3\times3\times3$ diamond-cubic supercells for both silicon and diamond hosts. Each motif is evaluated in charge states $q=+1,0,-1$. Structures are first relaxed using PBE with $\Gamma$-point sampling until the maximum residual force is below 0.03 eV$\cdot\rm{\mathring{A}^{-1}}$. 
To identify low-energy magnetic configurations, multiple total magnetic-moment sectors compatible with the electron count are independently relaxed for each charge state, with higher-spin sectors additionally considered for motifs containing transition metals. The lowest-energy magnetic configuration is retained for subsequent electronic structure calculations.
Electronic energies and defect states are then recomputed on the PBE-relaxed structures using the fully self-consistent Heyd-Scuseria-Ernzerhof (HSE06) screened hybrid functional \cite{krukau2006influence} to describe localized defect states and their positions relative to the host band edges. The pristine-host energies, valence- and conduction-band edges, and elemental reference chemical potentials entering the defect energetics are consistently evaluated at the HSE06 level. Charged-defect total energies are corrected for interactions between periodic images and the compensating background using the Freysoldt-Neugebauer-Van de Walle correction \cite{freysoldt2014first}, with static dielectric constants of 11.7 and 5.7 used for silicon and diamond, respectively.

Finally, the formation energy of defect $D$ in charge state $q$ is evaluated as
\begin{equation}
E_{\mathrm{f}}[D^q]=E[D^q]-E_{\mathrm{bulk}}-\sum_i n_i\mu_i+q(E_{\mathrm{VBM}}+E_{\mathrm{F}})+E_{\mathrm{corr}}^q,
\end{equation}
where $E[D^q]$ and $E_{\mathrm{bulk}}$ are the defect and pristine-host HSE06 energies, $n_i$ denotes the change in the number of atoms of species $i$, such that $n_i<0$ corresponds to atoms removed from the supercell and $n_i>0$ to atoms added to form the defect, and $\mu_i$ is the corresponding atomic chemical potential, $E_{\mathrm{F}}$ is measured from the aligned VBM, and $E_{\mathrm{corr}}^q$ is the finite-size correction. Thermodynamic CTLs between charge states $q_1$ and $q_2$ are obtained from crossings of the corresponding formation energies,
\begin{equation}
\varepsilon(q_1/q_2)=\frac{E[D^{q_1}]+E_{\mathrm{corr}}^{q_1}-E[D^{q_2}]-E_{\mathrm{corr}}^{q_2}}{q_2-q_1}-E_{\mathrm{VBM}}.
\end{equation}

\section*{Acknowledgments}
The authors thank Juno Nam and Jiajun He for helpful discussions. D.B. acknowledges support by funding from a Google PhD fellowship in Machine Learning and ML Foundations. M.S. acknowledges the financial support from Deutsche Forschungsgemeinschaft (DFG) through Grant No. CRC 1114, “Scaling Cascades in Complex Systems” (Project No. 235221301), Project No. B08, “Multiscale Boltzmann Generators”. The Flatiron Institute is a division of the Simons Foundation.

\newpage

\begingroup
\sffamily
\raggedright
\setlength{\parindent}{0pt}

\vspace*{-24pt}

{\bfseries\fontsize{18}{22}\selectfont
JANUS: A Multimodal Foundation Neural Sampler for Disordered Materials: Supplementary Information
\par}
\vspace{9pt}
{\bfseries\selectfont
Denis Blessing$^{1,2,*,\#}$, Mouyang Cheng$^{1,3,4,*,\#}$, Maximilian Schebek$^{5}$, Jutta Rogal$^{6}$, Mingda Li$^{3,7}$, Carles Domingo-Enrich$^{1,\dagger}$, and Yuanqi Du$^{1,\dagger}$
\par}
\vspace{12pt}
{\fontsize{9.5}{11.5}\selectfont
$^{1}$Microsoft Research New England, Cambridge, MA 02142, USA\\
$^{2}$Karlsruhe Institute of Technology, 76131 Karlsruhe, Germany\\
$^{3}$Center for Computational Science and Engineering, MIT, Cambridge, MA 02139, USA\\
$^{4}$Department of Materials Science and Engineering, MIT, Cambridge, MA 02139, USA\\
$^{5}$Department of Physics, Freie Universit\"{a}t Berlin, 14195 Berlin, Germany\\
$^{6}$Initiative for Computational Catalysis, Flatiron Institute, New York, NY 10010, USA\\
$^{7}$Department of Nuclear Science and Engineering, MIT, Cambridge, MA 02139, USA\\
$^{*}$These authors contributed equally.\\
$^{\#}$Work done during internship at Microsoft Research.\\
$^\dagger$Correspondence. Email: carlesd@microsoft.com, yuanqidu@microsoft.com
\par}
\endgroup

\addtocontents{toc}{\protect\setcounter{tocdepth}{2}}
\tableofcontents
\setcounter{figure}{0}
\setcounter{table}{0}
\setcounter{equation}{0}
\renewcommand{\thefigure}{S\arabic{figure}}
\renewcommand{\thetable}{S\arabic{table}}
\renewcommand{\theequation}{S\arabic{equation}}
\renewcommand{\theHfigure}{Supp\arabic{figure}}
\renewcommand{\theHtable}{Supp\arabic{table}}
\renewcommand{\theHequation}{Supp\arabic{equation}}

\newcommand{\enghist}[3]{%
  \begin{scope}[shift={(#1,#2)}]
    \pgfmathsetmacro{\sig}{0.28-0.18*(#3)}%
    \pgfmathsetmacro{\amp}{0.20+0.28*(#3)}%
    \pgfmathtruncatemacro{\mixp}{100-100*(#3)}%
    \draw[rounded corners=1pt,draw=black!45,line width=0.4pt,fill=white] (-0.40,-0.32) rectangle (0.40,0.32);
    \fill[red!\mixp!blue,opacity=0.20]
      plot[domain=-0.34:0.34,samples=34,smooth] (\x,{-0.24+\amp*exp(-(\x*\x)/(2*\sig*\sig))})
      -- (0.34,-0.24) -- (-0.34,-0.24) -- cycle;
    \draw[red!\mixp!blue,line width=0.8pt]
      plot[domain=-0.34:0.34,samples=34,smooth] (\x,{-0.24+\amp*exp(-(\x*\x)/(2*\sig*\sig))});
  \end{scope}%
}
\newtheorem{proposition}{Proposition}

\section{Method details}

\begin{figure}[!h]
 \centering
 \includegraphics[width=0.95\textwidth]{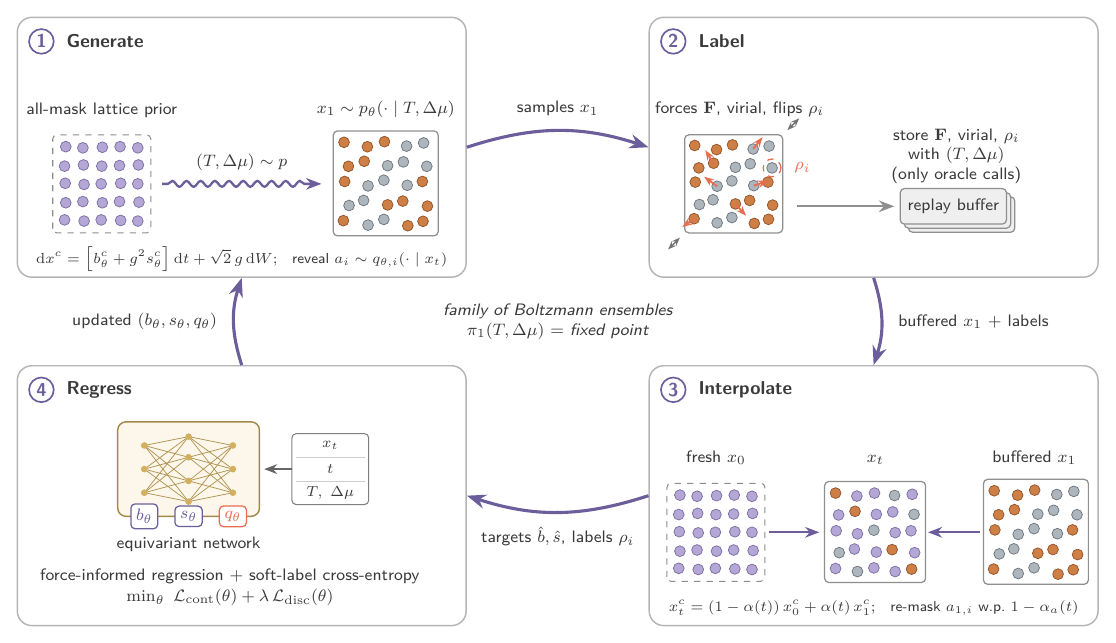}
 \caption{\textbf{Self-consistent fixed-point training of \name{}.} At each round, the current model generates terminal configurations by integrating the coupled dynamics---the discretized stochastic differential equation \eqref{eq:em} for displacements and log-volume, monotone unmasking \eqref{eq:reveal} for the species---from the all-mask lattice prior, under conditions $(T,\Delta\mu)$ drawn freshly per chain (\emph{generate}); the interatomic potential labels each terminal with the forces, the virial, and the heat-bath conditionals $\rho_i$ of Eq.~\eqref{eq:flip-score}, which are stored in a replay buffer together with the condition (\emph{label}); intermediate states $(\mathbf a_t,\mathbf x_t)$ are formed by interpolating the continuous channels between fresh prior draws and buffered terminals while re-masking the buffered species site-wise (\emph{interpolate}); and the network heads---the velocities and scores $(b_\theta,s_\theta)$ of the continuous channels and the species posterior $q_\theta$---are updated by force-informed regression and soft-label cross-entropy (\emph{regress}). The family of Boltzmann ensembles $\pi_1 (T,\Delta\mu)$ is a fixed point of this cycle.}
 \label{fig:m-fixedpoint}
\end{figure}

\subsection{Algorithmic description of training and inference}
\label{si:algorithms}

Algorithms~\ref{alg:train} and~\ref{alg:infer} describe the amortized training loop of the Methods section and the free-energy inference with forward--backward path weights. Two structural properties deserve emphasis. First, every evaluation of the interatomic potential happens in the \emph{label} step: each terminal configuration is stored together with its labels---the continuous scores \eqref{eq:cont-targets}, i.e.\ forces and virial (the derivative of the energy with respect to the isotropic log-volume), and the heat-bath conditionals \eqref{eq:flip-score} of \emph{all} $N$ sites, obtained from a single batched evaluation---so the inner gradient steps replay this oracle information many times, and the number of energy and force evaluations is set by the generation cadence alone, independently of the number of gradient steps. Second, both loops are conditional throughout: every chain carries its own $(T,\Delta\mu)$ which is stored in the buffer alongside the labels, and a replayed terminal is interpolated, targeted and regressed at exactly its stored condition, so a single network amortizes the whole conditioning window. The diffusion strength is chosen per channel and scales with temperature as $g_c(t;T)=g_c(t)\,(T/T_{\mathrm{ref}})^{1/2}$, matching the scaling of the thermal fluctuations that the prior widths \eqref{eq:prior-u} already follow.

\begin{algorithm}[t]
\caption{\name{} amortized fixed-point training}
\label{alg:train}
\begin{algorithmic}[1]
\State \textbf{input:} interatomic potential $U$ (forces, virial, single-site substitution energies); pressure $P$; condition distribution $p(T,\Delta\mu)$; schedules $\alpha(t),\alpha_a(t)$; diffusion strengths $g_c(t)$; prior parameters of Eqs.~\eqref{eq:prior-w}--\eqref{eq:prior-u}; loss weight $\lambda$
\State initialize the network $\theta$ with heads $(b^u_\theta,s^u_\theta,b^v_\theta,s^v_\theta,q_\theta)$ (scalar heads at zero)
\State fill the replay buffer $\mathcal B$ with conditioned prior draws
\For{a fixed number of outer rounds}
  \State \textbf{// generate}
  \State draw per-chain conditions $(T,\Delta\mu)\sim p$
  \State draw $x_0=(\mathbf a_0,\mathbf u_0,v_0)$: all-mask species, $(\mathbf u_0,v_0)$ from Eqs.~\eqref{eq:prior-w}--\eqref{eq:prior-u}
  \State integrate Eqs.~\eqref{eq:em}--\eqref{eq:reveal} for $M$ steps under $(T,\Delta\mu)$ to obtain terminals $x_1=(\mathbf a_1,\mathbf u_1,v_1)$
  \State \textbf{// label}
  \State evaluate the continuous scores \eqref{eq:cont-targets} and the heat-bath conditionals $\rho_i$ of Eq.~\eqref{eq:flip-score} for all sites 
  \State push $\big(x_1,\,\text{labels},\,(T,\Delta\mu)\big)$ onto $\mathcal B$
  \For{a fixed number of inner steps}
    \State draw $\big(x_1,\text{labels},(T,\Delta\mu)\big)\sim\mathcal B$; fresh $(\mathbf u_0,v_0)$ at the stored condition; $t\sim\mathcal U[0,1]$
    \State \textbf{// interpolate}
    \State $x^c_t\gets(1-\alpha(t))\,x^c_0+\alpha(t)\,x^c_1$ for $c\in\{u,v\}$;\quad mask each site of $\mathbf a_1$ independently w.p.\ $1-\alpha_a(t)$
    \State $\hat b^{\,c}\gets\dot\alpha(t)\,(x^c_1-x^c_0)$;\quad $\hat s^{\,c}\gets$ combination of stored target score and prior score, Eqs.~\eqref{eq:gtsi}--\eqref{eq:gtsi-weight}
    \State $\ell_i\gets\rho_i$ of Eq.~\eqref{eq:flip-score} at the masked sites
    \State \textbf{// regress}
    \State take a gradient step on $\mathcal L_{\mathrm{cont}}+\lambda\,\mathcal L_{\mathrm{disc}}$ [Eqs.~\eqref{eq:loss-cont} and \eqref{eq:loss-disc}] at the stored condition
  \EndFor
\EndFor
\State \textbf{return} $\theta$
\end{algorithmic}
\end{algorithm}

At inference (Algorithm~\ref{alg:infer}), one rollout produces a sample and its log-weight simultaneously: the continuous channels accumulate the forward--backward Gaussian log-ratio \eqref{eq:fbrnd-uw} step by step, the species channel accumulates the realized log-probability \eqref{eq:masked-qs} of its reveal sequence at the reveal moments, and the target enters once, at the terminal, through $\log\pi_1(x_M)$ in Eq.~\eqref{eq:fbrnd}.

\begin{algorithm}[t]
\caption{\name{} inference with forward--backward path weights}
\label{alg:infer}
\begin{algorithmic}[1]
\State \textbf{input:} trained network $\theta$; condition $(T,\Delta\mu)$ at pressure $P$; grid $t_n=nh$ with $h=1/M$; number of trajectories $n_{\mathrm{traj}}$
\For{$j=1,\dots,n_{\mathrm{traj}}$}
  \State draw $\mathbf a_0=(\texttt{M},\dots,\texttt{M})$ and $(\mathbf u_0,v_0)\sim\pi_0(\cdot\mid T,\Delta\mu)$
  \State $\log W\gets-\log\pi_0(\mathbf u_0,v_0)$;\quad $\log q_{\mathrm{disc}}\gets0$
  \For{$n=0,\dots,M-1$}
    \State evaluate the network once at $(x_n,t_n)$: fields $b^c_\theta,s^c_\theta$ and posteriors $q_{\theta,i}$
    \State advance $(\mathbf u,v)$ by the Euler--Maruyama step \eqref{eq:em}
    \State reveal each still-masked site w.p.\ $\gamma_n$ of Eq.~\eqref{eq:reveal} (with $\gamma_{M-1}=1$); a revealed site $i$ draws $a_{1,i}\sim q_{\theta,i}(\cdot\mid x_n)$
    \State $\log q_{\mathrm{disc}}\gets\log q_{\mathrm{disc}}+\textstyle\sum_{i\,\mathrm{revealed}}\log q_{\theta,i}(a_{1,i}\mid x_n)$
    \State $\log W\gets\log W+\Delta^{uv}_n$ of Eq.~\eqref{eq:fbrnd-uw}, with the backward drift of Eq.~\eqref{eq:bwd-sde} evaluated at $(x_{n+1},t_{n+1})$
  \EndFor
  \State $\log W\gets\log W+\log\pi_1(x_M)-\log q_{\mathrm{disc}}$
  \State store $\big(x_M^{(j)},\log W^{(j)}\big)$
\EndFor
\State \textbf{semi-grand potential:} $\log\widehat\Xi=\operatorname*{logsumexp}_j\log W^{(j)}-\log n_{\mathrm{traj}}$;\quad $\widehat\Phi=-k_{\mathrm B}T\log\widehat\Xi$ \Comment{Eq.~\eqref{eq:fbrnd-Z}}
\State \textbf{observables:} $\bar w_j=e^{\log W^{(j)}}\big/\sum_{j'}e^{\log W^{(j')}}$;\quad $\langle O\rangle=\sum_j\bar w_j\,O\big(x_M^{(j)}\big)$
\State \textbf{diagnostic:} $\mathrm{ESS}=\big(\sum_j e^{\log W^{(j)}}\big)^2\big/\sum_j e^{2\log W^{(j)}}$
\State \textbf{return} samples $\{x_M^{(j)}\}$, weights $\{\bar w_j\}$, and estimates $(\log\widehat\Xi,\widehat\Phi,\langle O\rangle,\mathrm{ESS})$
\end{algorithmic}
\end{algorithm}

\subsection{The discrete masked sampler and its Boltzmann fixed point}
\label{si:discrete-fixed-point}

This section formalizes the discrete formulation of the fixed-point iteration introduced in the Methods section and demonstrates that the Boltzmann ensemble is a fixed point of this process. For the underlying assumptions and theoretical guarantees regarding convergence, we refer the reader to He et al. \cite{he2026fixed}. Throughout we freeze the continuous channels and the thermodynamic conditions and write $\pi_1(\mathbf a)$ for the target's conditional law of the species vector $\mathbf a\in\mathcal A^N$; for the binary alloy $\mathcal A=\{0,1\}$, but nothing below depends on the alphabet size, so vacancies and multi-component systems are covered verbatim. The heat-bath conditionals of the target are
\begin{equation}
\rho_i(b\mid \mathbf a_{-i})=\pi_1\big(a_i=b\,\big|\,\mathbf a_{-i}\big),
\end{equation}
cf.\ Eq.~\eqref{eq:flip-score} of the Methods section. A \emph{partially revealed state} is a vector $z\in(\mathcal A\cup\{\texttt{M}\})^N$; we write $\mathcal R(z)$ and $\mathcal M(z)$ for its revealed and masked sites and say that a terminal state $\mathbf a$ is \emph{consistent} with $z$, written $\mathbf a\triangleright z$, if $a_i=z_i$ for all $i\in\mathcal R(z)$.

\paragraph{The sampler.} A model is a family of heads $q=\{q_i(\cdot\mid z)\}$, one categorical law per masked site of every partially revealed state. The reveal process draws i.i.d.\ reveal times $\tau_i$ with $\mathbb P(\tau_i\le t)=\alpha_a(t)$ and, at its reveal time, site $i$ draws its symbol from $q_i(\cdot\mid z)$ evaluated at the then-current state; we analyze the exact sequential process (one reveal at a time), to which the tau-leap integrator of the Methods section converges as the number of reveals per step approaches one. Since the $\tau_i$ are i.i.d., the reveal order is a uniformly random permutation $\sigma$ of the sites, independent of everything else, and the terminal law of the rollout is
\begin{equation}
\pi_q(\mathbf a)=\mathbb E_\sigma\Bigg[\prod_{k=1}^{N}q_{\sigma(k)}\big(a_{\sigma(k)}\,\big|\,z^{\sigma,\mathbf a}_{k-1}\big)\Bigg],
\label{eq:si-aoar}
\end{equation}
where $z^{\sigma,\mathbf a}_{k-1}$ is the state in which sites $\sigma(1),\dots,\sigma(k-1)$ show their values under $\mathbf a$ and all others are masked. Equation \eqref{eq:si-aoar} is the any-order autoregressive representation of the masked sampler \cite{ou2024your}.

\paragraph{The training map.} One idealized round of the iteration consists of the two maps
\begin{equation}
q\;\xrightarrow{\ \Lambda\ }\;\pi_q\;\xrightarrow{\ \Psi\ }\;q',
\end{equation}
Here, “idealized” means that the model can represent any collection of conditional distributions and that each training round finds the exact population-loss minimizer. We further assume that training uses only terminals generated by the current model, rather than the mixture of current and earlier rollout distributions retained by the practical replay buffer.
A model $q$ is a \emph{fixed point} if $\Psi(\Lambda(q))=q$ on all states reachable under its own rollout.

\begin{proposition}[Population minimizer]
\label{prop:si-minimizer}
Fix a buffer law $\pi$. For every partially revealed state $z$ with $\pi$-positive probability and every $i\in\mathcal M(z)$, the unique minimizer of the population loss \eqref{eq:loss-disc} is
\begin{equation}
q'_i(b\mid z)=\mathbb E_{\mathbf a\sim\pi}\big[\rho_i(b\mid\mathbf a_{-i})\,\big|\,\mathbf a\triangleright z\big].
\label{eq:si-minimizer}
\end{equation}
\end{proposition}

\noindent\emph{Proof.} The cross-entropy $-\sum_b \ell(b)\log q(b)$, averaged over the conditional law of the label $\ell$ given the input $z$, is minimized over categorical $q$ uniquely by the conditional mean $q=\mathbb E[\ell\mid z]$. Because every site is masked independently with a probability that does not depend on its symbol, the mask pattern is independent of $\mathbf a$; conditioning on the observed state $z$ is therefore equivalent to conditioning on the event $\{\mathbf a\triangleright z\}$, which gives \eqref{eq:si-minimizer}. \hfill$\square$

\begin{proposition}[The Boltzmann ensemble is a fixed point]
\label{prop:si-fixed-point}
Let $q^{\pi_1}$ denote the exact denoising posteriors of the target, $q^{\pi_1}_i(b\mid z)=\pi_1\big(a_i=b\mid \mathbf a_{\mathcal R(z)}=z_{\mathcal R(z)}\big)$, where $\mathbf a_{\mathcal R(z)}$ and $z_{\mathcal R(z)}$ denote the entries of $\mathbf a$ and $z$ at the unmasked sites of $z$. Then
(i) $\Psi(\pi_1)=q^{\pi_1}$, and (ii) $\Lambda(q^{\pi_1})=\pi_1$. Consequently $q^{\pi_1}$ is a fixed point of the iteration, and its rollout samples the Boltzmann ensemble exactly.
\end{proposition}

\noindent\emph{Proof.} (i) Apply Proposition~\ref{prop:si-minimizer} with $\pi=\pi_1$. Since $\rho_i(b\mid\mathbf a_{-i})$ does not depend on $a_i$, the tower property of conditional expectation gives
\begin{equation} \label{eq:si-fixed-point}
\mathbb E_{\pi_1}\big[\rho_i(b\mid\mathbf a_{-i})\,\big|\,\mathbf a\triangleright z\big]
=\sum_{\mathbf a_{\mathcal M(z)\setminus i}}\pi_1\big(\mathbf a_{\mathcal M(z)\setminus i}\,\big|\,\mathbf a_{\mathcal R(z)}=z_{\mathcal R(z)}\big)\,
\pi_1\big(a_i=b\,\big|\,\mathbf a_{-i}\big)
=\pi_1\big(a_i=b\,\big|\,\mathbf a_{\mathcal R(z)}=z_{\mathcal R(z)}\big),
\end{equation}
which is $q^{\pi_1}_i(b\mid z)$. In the first equality of \eqref{eq:si-fixed-point} we rewrote the conditional expectation as a summation and we leveraged that $\rho_i(b\mid\mathbf a_{-i}) = \pi_1\big(a_i=b\,\big|\,\mathbf a_{-i}\big)$ by definition. (ii) Conditional on any reveal order $\sigma$, the rollout with heads $q^{\pi_1}$ draws $a_{\sigma(1)},\dots,a_{\sigma(N)}$ sequentially from the exact conditionals of $\pi_1$; by the chain rule the resulting joint law is $\pi_1$ for \emph{every} $\sigma$, hence also after averaging over $\sigma$ in \eqref{eq:si-aoar}. \hfill$\square$

 \subsection{Fractional coordinates and the volume Jacobian}
\label{si:volume-jacobian}
The continuous channels are not Cartesian positions but fractional displacements $\mathbf u$ and the log-volume $v=\log V$. For the cubic cells used throughout, the Cartesian position of site $i$ is
\begin{equation}
\mathbf r_i=L\,(\mathbf q_i+\mathbf u_i),\qquad L=V^{1/3}=e^{v/3},
\label{eq:si-frac}
\end{equation}
where $\mathbf q_i$ is the fractional position of the reference lattice site. The energy and the density transform differently under this map. The energy is a function and is simply evaluated at the mapped positions, $U(\mathbf a,\mathbf u,v)\equiv U(\mathbf a,\{\mathbf r_i\})$; at fixed $\mathbf u$ it thereby depends explicitly on $v$, since every interatomic distance scales with $L$. The density is a measure and picks up the Jacobian $\mathrm d^{3N}\mathbf r=L^{3N}\,\mathrm d^{3N}\mathbf u=V^N\,\mathrm d^{3N}\mathbf u$. Integrating the cell volume with the scale-invariant measure $\mathrm dV/V=\mathrm dv$, the target \eqref{eq:target} in the sampled coordinates reads
\begin{equation}
\log\pi_1(\mathbf a,\mathbf u,v)=-\beta\big[U(\mathbf a,\mathbf u,v)+PV-\Delta\mu\,N_{\mathrm B}(\mathbf a)\big]+Nv+\mathrm{const}.
\label{eq:si-target-frac}
\end{equation}
with $\beta=1/k_{\mathrm B}T$. For a fixed cell, $Nv$ is a constant and cancels from every normalized density, density ratio and importance weight. Once $v$ is sampled, it varies between configurations and is precisely what makes the ensemble isobaric: it is the entropic contribution that balances the external pressure, and omitting it biases the volume distribution towards compressed cells.

\subsection{Generalized target score matching}
\label{si:gtsi}
The score target $\hat s^{\,c}$ of Eq.~\eqref{eq:cont-targets} is the target-endpoint form of the score identity \eqref{eq:tsi}. Its coefficient $1/\alpha(t)$ diverges as $t\to0$, i.e., precisely where no estimate is needed, since $p_0=\pi_{0,c}$ is known in closed form. Following prior work \cite{de2024target, blessing2026bridge}, we therefore regress on a convex combination of the score identities anchored at the two endpoints of the interpolant,
\begin{equation}
s^c(\mathbf a,\mathbf x,t)=\mathbb E\Big[\tfrac{\kappa(t)}{\alpha(t)}\,\nabla_{x_1^c}\log \pi_1(\mathbf a_1,\mathbf x_1) + \tfrac{1-\kappa(t)}{1-\alpha(t)}\,\nabla_{x_0^c}\log\pi_{0,c}(\mathbf{x}_0^c)\,\Big|\,(\mathbf a_t,\mathbf x_t)=(\mathbf a,\mathbf x)\Big],
\label{eq:gtsi}
\end{equation}
which holds for \emph{every} weight function $\kappa(t)$: conditionally on $(\mathbf a_t,\mathbf x_t)$, both inner terms are unbiased estimates of the marginal score: the first by the target score identity \eqref{eq:tsi}, the second by its mirror image with the roles of the endpoints exchanged, in which the integration by parts is carried out against the (Gaussian) prior rather than the target. The choice of $\kappa(t)$ therefore affects only the variance of the per-sample regression target, which is dominated by the squared coefficients $\kappa(t)^2/\alpha(t)^2$ and $(1-\kappa(t))^2/(1-\alpha(t))^2$. Balancing the two gives
\begin{equation}
    \kappa(t)=\frac{\alpha(t)^2}{\alpha(t)^2+\big(1-\alpha(t)\big)^2},
\label{eq:gtsi-weight}
\end{equation}
the variance-minimizing weight when the two endpoint scores carry comparable conditional second moments. With the linear schedule $\alpha(t)=t$ used in this work, $\kappa(t)=t^2/\big(t^2+(1-t)^2\big)$: the regression target moves smoothly from the exact prior score at $t=0$ to the pure force label \eqref{eq:cont-targets} at $t=1$, both coefficients remain bounded on all of $[0,1]$, and no clamping or time cutoff is required. The prior scores are closed-form throughout, $\nabla_{\mathbf u_0}\log\pi_{0,u}=-\mathbf u_0/\sigma_u(T)^2$ and $\partial_{v_0}\log\pi_{0,v}=-(v_0-\bar v)/\sigma_v^2$ from Eqs.~\eqref{eq:prior-w}--\eqref{eq:prior-u}.

\subsection{General proposals and tilted inference}
\label{si:steering-general}
The steering scheme presented in the Methods section biases only the continuous channels of the proposal; the species channel is propagated by the pretrained kernel and steered through the weights and the resampling alone. Here we outline the general construction in which the species proposal is biased toward the reward as well. The proposal again mirrors the pretrained integrator: each continuous channel advances with an arbitrary drift $\tilde b^c$ with Gaussian transition density $\mathcal N\big(x^c_n+\tilde b^c(\mathbf a_n,\mathbf x_n,t_n)\,h,\;2g^2(t_n)h\big)$, while the species channel keeps the reveal times of the pretrained sampler and lets a site revealed at step $n$ draw its symbol from an arbitrary categorical $q^{\mathrm{prop}}_{n,i}$ with full support. The log-weight update then acquires, in addition to the reward increment and the Gaussian density ratio, one realized-token ratio per site revealed in that step,
\begin{equation}
\begin{aligned}
\Delta\log w_n&=r_{t_{n+1}}(\mathbf a_{n+1},\mathbf x_{n+1})-r_{t_n}(\mathbf a_n,\mathbf x_n)+\sum_{i\in\mathcal R_n}\log\frac{q_{\theta,i}(a_{n+1,i}\mid \mathbf a_n,\mathbf x_n)}{q^{\mathrm{prop}}_{n,i}(a_{n+1,i})}\\
&\quad+\sum_{c\in\{u,v\}}\log\frac{\mathcal N\big(x^c_{n+1};\;x^c_n+\big[b^c_\theta+g^2(t_n)\,s^c_\theta\big](\mathbf a_n,\mathbf x_n,t_n)\,h,\;2g^2(t_n)h\big)}{\mathcal N\big(x^c_{n+1};\;x^c_n+\tilde b^c(\mathbf a_n,\mathbf x_n,t_n)\,h,\;2g^2(t_n)h\big)},
\end{aligned}
\label{eq:smc-weight-general}
\end{equation}
where $\mathcal R_n$ is the set of sites revealed in step $n$. A natural species proposal is \emph{tilted decoding} \cite{he2026rne,hasan2026discrete}: a site revealed at step $n$ draws its symbol from the reweighted categorical
\begin{equation}
q^{\mathrm{prop}}_{n,i}(b)\;\propto\;q_{\theta,i}(b\mid \mathbf a_n,\mathbf x_n)\,e^{\lambda_a\,r_{t_n}(\mathbf a^{(i\to b)}_n,\,\mathbf x_n)},
\label{eq:tilted-decoding}
\end{equation}
where $\mathbf a^{(i\to b)}_n$ denotes the species state with site $i$ revealed as symbol $b$ and $\lambda_a\ge0$ controls the strength of the discrete guidance. With $\lambda_a=0$ the token ratios vanish and the construction reduces to the continuous-guidance scheme of the Methods section.

\subsection{Derivation of the inference-time steering weights}
\label{si:steering}

This section derives the steering weights---the general form \eqref{eq:smc-weight-general} with both channels proposed freely, of which the continuous-guidance weight \eqref{eq:smc-weight} of the Methods section is the special case $q^{\mathrm{prop}}_{n,i}=q_{\theta,i}$. The derivation is the discrete-time counterpart, on the hybrid state space, of the continuous-time Radon--Nikodym construction of Ref.~\cite{he2026rne}. Throughout we write $x=(\mathbf a,\mathbf x)$ for the full state, as in the main text, and $x_{0:M}$ for a trajectory on the grid $t_n=nh$, $h=1/M$.

Conditionally on the current state $x_n$, the implemented sampler advances the continuous channels by independent Gaussian steps and each species site independently, so its one-step transition kernel is the exactly evaluable product
\begin{equation}
\overrightarrow{p}_n(x_{n+1}\mid x_n)
=\prod_{c\in\{u,v\}}\mathcal N\!\Big(x^c_{n+1};\;x^c_n+\big[b^c_\theta+g^2(t_n)\,s^c_\theta\big](x_n,t_n)\,h,\;2g^2(t_n)h\,I\Big)\;
\prod_{i=1}^{N}\overrightarrow{p}^{\,a}_{n,i}\big(a_{n+1,i}\mid a_{n,i};x_n\big),
\label{eq:si-kernel}
\end{equation}
with the per-site species factor determined by the reveal probability $\gamma_n$ of Eq.~\eqref{eq:reveal} and the head of Eq.~\eqref{eq:masked-head},
\begin{equation}
\overrightarrow{p}^{\,a}_{n,i}(b'\mid b;x_n)=
\begin{cases}
\gamma_n\,q_{\theta,i}(b'\mid x_n), & b=\texttt{M},\ b'\in\mathcal A,\\[2pt]
1-\gamma_n, & b=b'=\texttt{M},\\[2pt]
\mathbf 1\{b'=b\}, & b\in\mathcal A\ \text{(revealed sites are frozen)}.
\end{cases}
\label{eq:si-species-kernel}
\end{equation}
The proposal kernel $\overrightarrow{q}_n$ has the same form, with an arbitrary drift $\tilde b^c$ in place of $b^c_\theta+g^2 s^c_\theta$ and an arbitrary full-support categorical $q^{\mathrm{prop}}_{n,i}$ in place of $q_{\theta,i}$, the reveal probabilities being shared. Let $\mathbb Q(x_{0:M})=\pi_0(x_0)\prod_{n}\overrightarrow{q}_n(x_{n+1}\mid x_n)$ be the law of the proposal chain, where $\pi_0$ is the joint prior.

For any unnormalized terminal density $\tilde p(x)$ and any normalized backward kernels $\overleftarrow{c}_n$, the weight
\begin{equation}
w(x_{0:M})=\frac{\tilde p(x_M)}{\pi_0(x_0)}\,\prod_{n=0}^{M-1}\frac{\overleftarrow{c}_n(x_n\mid x_{n+1})}{\overrightarrow{q}_n(x_{n+1}\mid x_n)}
\label{eq:si-generic-w}
\end{equation}
satisfies $\mathbb E_{\mathbb Q}\big[w\,O(x_M)\big]=\int\tilde p\,O\,\mathrm dx$ for any observable $O$, as the proposal density cancels and $x_0,\ldots,x_{M-1}$ integrate out in turn since each $\overleftarrow{c}_n$ is normalized. Nothing in this argument refers to the nature of the state space: the integral $\int\mathrm dx$ is the sum over species configurations combined with the Lebesgue integral over the continuous channels, and the successive marginalizations interleave sums and Gaussian integrals, so the identity holds verbatim on the hybrid state space. 
The choice of $\overleftarrow{c}_n$ therefore affects only the variance. 

The free-energy path weight \eqref{eq:fbrnd} is the case $\tilde p=\pi_1$, $\overrightarrow{q}_n=\overrightarrow{p}_n$, and $\overleftarrow{c}_n$ the time reversal of the learned chain---the backward Gaussian kernels of \eqref{eq:fbrnd-uw}, built from the reversed SDE \eqref{eq:bwd-sde}, for the continuous channels and the parameter-free re-masking kernel for the species. Note that \eqref{eq:fbrnd} should in principle contain a forward--backward ratio for the kernels corresponding to the masking-remasking schedule: under the backward process, a revealed site is masked with probability
$[\alpha_a(t_{n+1})-\alpha_a(t_n)]/\alpha_a(t_{n+1})$ and masked sites remain masked. However, the forward and backward
schedule factors telescope over the trajectory, leaving
$-\log q_{\mathrm{disc}}$ as the sole discrete contribution.

For steering, let $p_{\theta,n}$ denote the (intractable) marginals of the pretrained chain, with $p_{\theta,0}=\pi_0$ and $p_{\theta,M}=p_{\text{base}}$. Choosing the tilted model ensemble $\tilde p=p_{\text{base}}\,e^{\eta r}$ and the exact time reversal of the pretrained chain,
\begin{equation}
\overleftarrow{c}_n(x_n\mid x_{n+1})=\frac{\overrightarrow{p}_n(x_{n+1}\mid x_n)\,p_{\theta,n}(x_n)}{p_{\theta,n+1}(x_{n+1})},
\label{eq:si-reversal}
\end{equation}
which is normalized by construction, the intractable marginals telescope against $p_{\text{base}}$ and the prior, and \eqref{eq:si-generic-w} reduces to
\begin{equation}
\log w(x_{0:M})=\eta\,r(x_M)+\sum_{n=0}^{M-1}\log\frac{\overrightarrow{p}_n(x_{n+1}\mid x_n)}{\overrightarrow{q}_n(x_{n+1}\mid x_n)},
\label{eq:si-steering-w}
\end{equation}
free of intractable quantities. Writing $\eta\,r(x_M)=\sum_{n}\big[r_{t_{n+1}}(x_{n+1})-r_{t_n}(x_n)\big]$, which uses only the boundary conditions $r_0\equiv0$ and $r_1=\eta r$, yields the incremental update \eqref{eq:smc-weight-general}. The intermediate reward hence cancels from the terminal weight; it matters only through resampling, as the population after step $n$ approximates $\propto p_{\theta,n}\,e^{r_{t_n}}$, so $r_t$ controls how gradually the tilt is introduced along the rollout. The weights are exact for any step size $h$, and any guidance strengths.

\paragraph{The kernel ratio splits by channel.} Because base and proposal share the reveal probabilities, the pattern factors $\gamma_n$, $1-\gamma_n$ and the frozen-site factors in \eqref{eq:si-species-kernel} cancel between $\overrightarrow{p}_n$ and $\overrightarrow{q}_n$, and because both Gaussian factors share the variance $2g^2(t_n)h$, their normalizers cancel as well. Writing $\bar b^c_n=[b^c_\theta+g^2(t_n)\,s^c_\theta](x_n,t_n)$ for the pretrained drift and $\tilde b^c_n$ for the proposal drift, the kernel ratio in \eqref{eq:si-steering-w} reduces to
\begin{equation}
\log\frac{\overrightarrow{p}_n}{\overrightarrow{q}_n}
=\sum_{c\in\{u,v\}}\frac{\big\|x^c_{n+1}-x^c_n-\tilde b^c_n h\big\|^2-\big\|x^c_{n+1}-x^c_n-\bar b^c_n h\big\|^2}{4g^2(t_n)h}
\;+\sum_{i\in\mathcal R_n}\Big[\log q_{\theta,i}\big(a_{n+1,i}\mid x_n\big)-\log q^{\mathrm{prop}}_{n,i}\big(a_{n+1,i}\big)\Big],
\label{eq:si-kernel-ratio}
\end{equation}
where $\mathcal R_n$ is the set of sites revealed in step $n$ and both categoricals are evaluated at the realized symbols; every term is computed along the trajectory the sampler generates anyway, at no additional network cost. Substituting \eqref{eq:si-kernel-ratio} into \eqref{eq:si-steering-w} yields exactly the general steering weight \eqref{eq:smc-weight-general} of Section~\ref{si:steering-general}; propagating the species channel by the pretrained head, $q^{\mathrm{prop}}_{n,i}=q_{\theta,i}$, cancels the token ratios and recovers the continuous-guidance weight \eqref{eq:smc-weight} of the Methods section. With $\lambda_c=\lambda_a=0$ the two kernels coincide, the ratio vanishes identically, and the weight update consists of reward increments alone---pure reweighting of the unmodified sampler.

\subsection{Choices of intermediate reward and guidance}
\label{si:steering-choices}
The choices of the intermediate reward and proposal dynamics affect only the weight variance. In this work, we use two different variants outlined below.

\paragraph{Annealed reward.} When $r$ is meaningful on partially generated states, the annealed reward~\cite{cheng2026atlas} $r_t(x)=\alpha_r(t)\,\eta\,r(x)$ scores the current state directly under a ramp $\alpha_r(0)=0$, $\alpha_r(1)=1$. On the hybrid state this requires a convention for the species argument: revealed sites enter with their committed symbols, and still-masked sites are soft-completed by the head marginals $q_{\theta,i}$.

\paragraph{Tweedie reward.} When $r$ is defined only on clean terminal structures, we score the per-channel posterior means of the terminal state. For a continuous channel, the interpolant \eqref{eq:cont-interpolant} is deterministic given its endpoints, so conditioning it and the velocity \eqref{eq:cont-velocity} on $x_t=x$ gives two linear equations for the posterior means $m^c_k=\mathbb E[x^c_k\mid x_t=x]$, namely $x^c=(1-\alpha)m^c_0+\alpha m^c_1$ and $b^c_\theta=\dot\alpha\,(m^c_1-m^c_0)$ \cite{cheng2026atlas}, and hence
\begin{equation}
\hat x^c_1(x,t)=x^c+\frac{1-\alpha(t)}{\dot\alpha(t)}\,b^c_\theta(x,t),
\label{eq:si-tweedie}
\end{equation}
the posterior-mean estimate available from the velocity head alone \cite{robbins1992empirical}. For the species channel no estimate is needed: the head \emph{is} the posterior over the terminal identity, and the soft species argument
\begin{equation}
\hat\rho_i(b\mid x,t)=
\begin{cases}
\mathbf 1\{b=a_i\}, & i\ \text{revealed},\\
q_{\theta,i}(b\mid x), & i\ \text{masked},
\end{cases}
\label{eq:si-tweedie-species}
\end{equation}
is exact on revealed sites, where the terminal symbol is already committed.

\subsection{Grand-canonical interpretation of the vacancy-token representation}
JANUS represents the defect system using a fixed number $N$ of lattice sites (the same symbol as in the main text), each carrying either an atomic-species token or a vacancy token. Although $N$ is fixed, the number of physical atoms fluctuates according to $N_{\mathrm{atom}}=N-N_{\mathrm{VAC}}$.
Thus, fluctuations of the vacancy population correspond directly to fluctuations of the physical particle number.

For an occupied site containing atomic species $a$, the configurational grand-canonical measure contains the single-particle factor $V/{\Lambda_a^3}$, where $\Lambda_a=\sqrt{2\pi\hslash^2/(m_a k_{\mathrm B}T)}$ is the thermal de Broglie wavelength. Replacing an atom of species $a$ by a vacancy removes this factor and therefore contributes $-\log(\frac{V}{\Lambda_a^3})$ to the logarithm of the target probability.
Because $T$ and $V$ are fixed in the present $\mu VT$ defect calculations, this contribution is constant per vacancy and can be absorbed into an effective vacancy chemical potential,
\begin{equation}
\mu_{\mathrm{VAC}}^{\mathrm{eff}}=\mu_{\mathrm{VAC}}-k_{\mathrm B}T\log\left(\frac{V}{\Lambda_a^3}\right).
\end{equation}
Here, $a$ denotes the atomic species relative to which the vacancy is defined, corresponding to Si for the silicon host and C for the diamond host. More generally, if vacancies are exchanged relative to multiple atomic species, the species-dependent thermal-wavelength factors can equivalently be absorbed into the corresponding atomic chemical potentials.

The equivalence to a grand-canonical description can also be seen directly from the chemical-potential contribution. Since the total number of lattice sites satisfies $N=N_{\mathrm{VAC}}+\sum_a N_a$, we can write
\begin{equation}
\sum_a \mu_a N_a+\mu_{\mathrm{VAC}}N_{\mathrm{VAC}}=\mu_{\mathrm{VAC}}N+\sum_a\left(\mu_a-\mu_{\mathrm{VAC}}\right)N_a.
\end{equation}
The first term is constant because $N$ is fixed and therefore cancels upon normalization of the Boltzmann distribution. The statistical weights consequently depend only on the chemical-potential differences $\mu_a-\mu_{\mathrm{VAC}}$. Thus, treating VAC as an additional lattice token is equivalent to coupling the physical atomic populations to particle reservoirs through chemical potentials while retaining a fixed-dimensional representation for neural-network sampling.

The continuous displacement associated with a vacancy token is an auxiliary variable rather than a physical atomic coordinate. To maintain the same dimensionality for all token configurations, vacancy displacements are assigned a normalized site-local Gaussian density $\rho_g(\mathbf{u}_i)$. The corresponding contribution to the target distribution is $\prod_{i:a_i=\mathrm{VAC}}\rho_g(\mathbf{u}_i)$.
Vacancy sites are excluded from the MACE interaction graph, so these auxiliary coordinates contribute no physical interaction energy. Because $\rho_g$ is normalized, integrating over all vacancy coordinates contributes unity and leaves the physical grand-canonical distribution unchanged.

Combining these terms gives the fixed-dimensional target distribution used for JANUS defect sampling,
\begin{equation}
\log \pi_1(\mathbf{a},\mathbf{u})=\beta\left[\sum_i \mu_{a_i}-U(\mathbf{a},\mathbf{u})-V_{\mathrm{cap}}(\mathbf{a})\right]+\sum_{i:a_i=\mathrm{VAC}}\log \rho_g(\mathbf{u}_i),
\end{equation}
where the chemical potentials are understood to include the constant measure factors described above. The vacancy-token construction therefore changes only the representation of the grand-canonical system and does not alter the physical Boltzmann weights after marginalization over the auxiliary vacancy coordinates.

\section{Free energy estimation}
\label{si:free-energy}
This section details how composition-resolved free energies, mixing free energies and phase boundaries are obtained from the sampler and from the reference Monte Carlo simulations. The semi-grand decomposition into fixed-composition partition functions $Z_n$, the composition-resolved Gibbs free energy $G(n)=-k_{\mathrm B}T\log Z_n$ and the mixing free energy $G_{\mathrm{mix}}$ are defined in Eqs.~\eqref{eq:xi-decomp} and~\eqref{eq:fmix} of the Methods.

\subsection{Phase coexistence from the composition-resolved free energy}
\label{si:fe-coexistence}
Phase coexistence is encoded in $G$ through the common-tangent construction: at coexistence the two phases share the exchange chemical potential and the tangent,
\begin{equation}
G'(n_{\mathrm I})=G'(n_{\mathrm{II}})=\Delta\mu^{*}(T),\qquad
G(n_{\mathrm{II}})-G(n_{\mathrm I})=\Delta\mu^{*}(T)\,(n_{\mathrm{II}}-n_{\mathrm I}),
\label{eq:tangent}
\end{equation}
so the coexisting compositions $x_{\mathrm I}(T),x_{\mathrm{II}}(T)$ of phases I and II are the tangency points of the double tangent and its slope is the coexistence potential $\Delta\mu^{*}(T)$; repeating the construction along isotherms traces the binodal, with convex $G$ signaling complete miscibility. Two complementary routes lead from the sampler to $G(n)$, distinguished by which part of the curve they can reach (Sections~\ref{si:legendre} and~\ref{si:fixed-comp}).


\subsection{Legendre--Fenchel transform of the semi-grand potential}
\label{si:legendre}

At fixed temperature and pressure, the semi-grand potential is related
to the fixed-composition free energies by
\begin{equation}
\Phi(T,\Delta\mu)
=
-k_{\mathrm B}T
\log\sum_{n=0}^{N}
\exp\!\left[
-\frac{G(n)-\Delta\mu\,n}{k_{\mathrm B}T}
\right].
\label{eq:si-phi-finite}
\end{equation}
For a finite simulation cell, this is a log-sum-exp relation rather
than an exact Legendre--Fenchel transform. We nevertheless construct
the convex approximation
\begin{equation}
G_{\mathrm{LF}}(n)
=
\sup_{\Delta\mu}
\big[\Delta\mu\,n+\Phi(T,\Delta\mu)\big].
\label{eq:legendre-hull}
\end{equation}
This construction does not generally recover the exact finite-cell
free energy $G(n)$, even when $G(n)$ is convex, because the semi-grand
partition function includes contributions from multiple compositions.

The finite-cell correction can be expressed through the semi-grand
composition probability
\begin{equation}
P_{\Delta\mu}(n)
=
\frac{
\exp\!\left[-\big(G(n)-\Delta\mu\,n\big)/k_{\mathrm B}T\right]
}{
\Xi(T,\Delta\mu)
},
\end{equation}
which, using that $\Phi(T,\Delta\mu)=-k_{\mathrm B}T\log\Xi$ and rearranging, yields
$G(n)
=
\Phi(T,\Delta\mu)+\Delta\mu\,n
-k_{\mathrm B}T\log P_{\Delta\mu}(n)$.
Thus, the Legendre--Fenchel construction omits a finite-cell
composition-fluctuation contribution. In the thermodynamic limit,
the corresponding per-atom construction recovers the lower convex
envelope of the limiting free energy; it cannot resolve non-convex
finite-cell free-energy profiles within a miscibility gap.

For Cu--Ni, we use Eq.~\eqref{eq:legendre-hull} as an approximate
finite-cell reconstruction of the convex free energy.
A single amortized model provides estimates of
$\Phi(T,\Delta\mu)$ through Eq.~\eqref{eq:fbrnd-Z} over a grid of
chemical potentials. At each composition, we maximize
$\Phi(T,\Delta\mu)+\Delta\mu\,n$ over this grid and subtract the
pure-component chord as in Eq.~\eqref{eq:fmix} to obtain the
approximate mixing free energy shown in Fig.~\ref{fig2}e.
The finite-cell approximation and the finite chemical-potential
grid introduce errors distinct from those of the partition-function
estimator. For Cu--Ag and Ni--Cr, we instead reconstruct the
composition-resolved free energies using the fixed-composition
route described below.

\subsection{Fixed-composition sampling}
\label{si:fixed-comp}
To resolve composition-dependent free-energy profiles within a
miscibility gap without relying on the Legendre--Fenchel
approximation, we train the sampler at fixed composition.
Fixed-composition sampling also avoids the poor coverage that semi-grand sampling provides for compositions suppressed by phase coexistence.
During the self-bootstrapped iteration the rollouts are generated by fixed composition sampling: the reveal steps of the species channel are restricted so that exactly $n$ sites carry species B, with the conditioning features set to the rung composition $c_0=n/N$; the continuous channels, the replay buffer and the training labels are unchanged. Because the restriction changes the reveal probabilities, the realized (restricted) categorical probabilities are the ones that enter $\log q_{\mathrm{disc}}$ in Eq.~\eqref{eq:masked-qs}, so that the path weight of Eq.~\eqref{eq:fbrnd} remains a valid importance weight for the canonical rung ensemble. The fixed-point argument of Section~\ref{si:discrete-fixed-point} is formulated for the unconstrained semi-grand iteration and does not directly cover this constrained variant; any residual bias of the trained sampler is corrected by the path weights. One amortized model thereby learns the \emph{canonical} (fixed-composition) ensembles of every rung $n=0,\dots,N$ across the temperature window and generates any rung ensemble at inference cost. The free energy follows from the same path-weight machinery as $\Xi$: applied to a constrained rollout, whose terminal law is the canonical rung ensemble, the forward--backward weight of Eq.~\eqref{eq:fbrnd} normalizes to $Z_n$ instead of $\Xi$ (up to the known tilt $e^{\Delta\mu\,n/k_{\mathrm B}T}$), so Eq.~\eqref{eq:fbrnd-Z} evaluates $G(n)=-k_{\mathrm B}T\log Z_n$ rung by rung from the model's own trajectories. In practice we sharpen this read by joining the generated ensembles of neighboring rungs through exact single-site substitution energies, at the cost of two additional batched energy evaluations per draw; the estimator is detailed in Section~\ref{si:ladder-bar}, and the canonical Monte Carlo composition ladder that supplies the reference values is described in Section~\ref{si:ref-data}. Beyond reaching inside the gap, the ladder is what makes the phase-diagram read \emph{self-contained}: applying Eq.~\eqref{eq:tangent} to the model's own $G(n)$ delivers the coexistence potential $\Delta\mu^{*}(T)$ together with the coexisting compositions. All Cu--Ag and Ni--Cr free energies and binodals in this work are computed through this route.

\subsection{Free energies along the composition ladder}
\label{si:ladder-bar}

Here, we provide details on the computation of the composition-resolved free energy $G(n)$. Additionally, we present a different estimator that replaces the direct path-weight estimation of $Z_n$ outlined in Section~\ref{si:fixed-comp} and only measures \emph{ratios} of neighboring rungs, in which the dominant share of the path-weight variance (the absolute free energy of the rung ensemble) cancels. It is the estimator used for all reported ladder results.

\paragraph{Substitution identity.}
For a configuration $\mathbf a$ on rung $n$, and an A site $i$, define the
forward substitution energy
\begin{equation}
\Delta U_i^{\mathrm{A}\to\mathrm{B}}=U(\mathbf a^{(i\to\mathrm B)},\mathbf u,v)-U(\mathbf a,\mathbf u,v).
\end{equation}
Similarly, for a configuration on rung $n+1$ and a B site $j$, define
\begin{equation}
\Delta U_j^{\mathrm{B}\to\mathrm{A}}=U(\mathbf a^{(j\to\mathrm A)},\mathbf u,v)-U(\mathbf a,\mathbf u,v).
\end{equation}
In both substitutions, the displacements and volume are held fixed. Every configuration on rung $n+1$ has exactly $n+1$ preimages obtained by choosing one of its B sites and retyping it as A. Conversely, every configuration on rung $n$ has $N-n$ possible A-to-B substitutions.
Summing over these configuration--site pairs and regrouping by the substituted configuration gives
\begin{equation}
\begin{aligned}
\frac{Z_{n+1}}{Z_n}&=\frac{N-n}{n+1}\left\langle\frac{1}{N-n}\sum_{i:a_i=\mathrm A}e^{-\beta\Delta U_i^{\mathrm{A}\to\mathrm{B}}}\right\rangle_n =\frac{N-n}{n+1}\left\langle\frac{1}{n+1}\sum_{j:a_j=\mathrm B}e^{-\beta\Delta U_j^{\mathrm{B}\to\mathrm{A}}}\right\rangle_{n+1}^{-1}.
\end{aligned}
\label{eq:ladder}
\end{equation}
\noindent\emph{Proof sketch.}
Let $\Omega_n=\{\mathbf a:N_{\mathrm B}(\mathbf a)=n\}$
denote the set of species configurations on rung $n$, and write the
fixed-composition partition function as
\[
Z_n
=
\sum_{\mathbf a\in\Omega_n}
\int d\lambda(\mathbf u,v)\,
e^{-\beta[U(\mathbf a,\mathbf u,v)+PV(v)]},
\]
where $d\lambda(\mathbf u,v)$ denotes the joint continuous measure. The measure is the same on
neighboring rungs because an A--B substitution does not change the total
number of atoms. 
For the forward direction, expanding the rung-$n$ expectation gives
\begin{align}
&\left\langle
\frac{1}{N-n}
\sum_{i:a_i=\mathrm A}
e^{-\beta\Delta U_i^{\mathrm A\to\mathrm B}}
\right\rangle_n
=
\frac{1}{(N-n)Z_n}
\sum_{\mathbf a\in\Omega_n}
\sum_{i:a_i=\mathrm A}
\int d\lambda(\mathbf u,v)\,
e^{-\beta[U(\mathbf a^{(i\to\mathrm B)},\mathbf u,v)+PV(v)]}.
\label{eq:si-ladder-forward-proof}
\end{align}
Every configuration $\mathbf a'\in\Omega_{n+1}$ appears exactly $n+1$
times in this sum, because any one of its $n+1$ B sites can be identified
as the site introduced by the substitution. Hence,
\[
\left\langle
\frac{1}{N-n}
\sum_{i:a_i=\mathrm A}
e^{-\beta\Delta U_i^{\mathrm A\to\mathrm B}}
\right\rangle_n
=
\frac{n+1}{N-n}\frac{Z_{n+1}}{Z_n}
\implies
\frac{Z_{n+1}}{Z_n}
=
\frac{N-n}{n+1}
\left\langle
\frac{1}{N-n}
\sum_{i:a_i=\mathrm A}
e^{-\beta\Delta U_i^{\mathrm A\to\mathrm B}}
\right\rangle_n.
\]

Similarly, expanding the reverse expectation gives
\begin{align}
&\left\langle
\frac{1}{n+1}
\sum_{j:a_j=\mathrm B}
e^{-\beta\Delta U_j^{\mathrm B\to\mathrm A}}
\right\rangle_{n+1}
=
\frac{1}{(n+1)Z_{n+1}}
\sum_{\mathbf a'\in\Omega_{n+1}}
\sum_{j:a'_j=\mathrm B}
\int d\lambda(\mathbf u,v)\,
e^{-\beta[U(\mathbf a'^{(j\to\mathrm A)},\mathbf u,v)+PV(v)]}.
\label{eq:si-ladder-reverse-proof}
\end{align}
After regrouping by the resulting configuration in $\Omega_n$, every
rung-$n$ configuration appears exactly $N-n$ times, once for each of its
A sites. Therefore,
\[
\left\langle
\frac{1}{n+1}
\sum_{j:a_j=\mathrm B}
e^{-\beta\Delta U_j^{\mathrm B\to\mathrm A}}
\right\rangle_{n+1}
=
\frac{N-n}{n+1}\frac{Z_n}{Z_{n+1}} \implies
\frac{Z_{n+1}}{Z_n}
=
\frac{N-n}{n+1}
\left\langle
\frac{1}{n+1}
\sum_{j:a_j=\mathrm B}
e^{-\beta\Delta U_j^{\mathrm B\to\mathrm A}}
\right\rangle_{n+1}^{-1},
\]
which proves both identities in Eq.~\eqref{eq:ladder}.
\hfill$\square$

Equation \eqref{eq:ladder} provides two one-sided estimators of the same neighboring-rung free-energy difference. We record their discrepancy as an overlap diagnostic, but use Bennett’s acceptance ratio \cite{bennett1976efficient} to combine the forward and reverse substitution works into the reported two-sided estimate, as described below.

\paragraph{Reduced works and Bennett acceptance ratio.}
To combine the two directions, augment each rung ensemble with a uniformly
chosen eligible substitution site. Define the forward-oriented reduced works
\begin{equation}
\begin{aligned}
w_{\mathrm f,i}
=
\beta\Delta U_i^{\mathrm{A}\to\mathrm{B}}
-\log\frac{N-n}{n+1},
&& \text{on } \mathbf a\sim\pi_n,
\qquad\quad
w_{\mathrm r,j}
=
-\beta\Delta U_j^{\mathrm{B}\to\mathrm{A}}
-\log\frac{N-n}{n+1},
&& \text{on } \mathbf a\sim\pi_{n+1}.
\end{aligned}
\label{eq:si-works}
\end{equation}
Both $w_{\mathrm f}$ and $w_{\mathrm r}$ represent the same forward reduced
potential difference, but are evaluated using samples from opposite ends of
the edge.

Let $n_{\mathrm f}$ and $n_{\mathrm r}$ denote the numbers of sampled configurations from rungs $n$ and $n+1$, respectively,
let
$C=\log(n_{\mathrm r}/n_{\mathrm f})$, and write $\mathrm{sigmoid}(z)=(1+e^{-z})^{-1}$ for the logistic function. The dimensionless edge free energy
\begin{equation}
\varphi_n
=
\beta\big[G(n+1)-G(n)\big]
=
-\log\frac{Z_{n+1}}{Z_n}
\end{equation}
is the solution of
\begin{equation}
\sum_{k=1}^{n_{\mathrm f}}
\frac{1}{N-n}
\sum_{i:\,a^{(k)}_i=\mathrm A}
\mathrm{sigmoid}\!\left(\varphi_n+C-w_{\mathrm f,i}^{(k)}\right)
+
\sum_{k=1}^{n_{\mathrm r}}
\frac{1}{n+1}
\sum_{j:\,a^{(k)}_j=\mathrm B}
\mathrm{sigmoid}\!\left(\varphi_n+C-w_{\mathrm r,j}^{(k)}\right)
=
n_{\mathrm r}.
\label{eq:si-bar}
\end{equation}
This is the standard Bennett acceptance-ratio equation rewritten using
 $\mathrm{sigmoid}(-z)=1-\mathrm{sigmoid}(z)$. Its left-hand side is monotone in $\varphi_n$, so
the unique finite root can be obtained robustly by bisection. BAR combines
the forward and reverse work distributions and is asymptotically optimal
under the standard independent-sampling assumptions. The discrepancy
between the two one-sided estimates in Eq.~\eqref{eq:ladder} is additionally
reported as an overlap diagnostic.

Accumulating the edge differences determines the free-energy curve relative
to the pure-A rung,
\begin{equation}
G(n)-G(0)
=
k_{\mathrm B}T\sum_{m=0}^{n-1}\varphi_m.
\end{equation}
The unknown value $G(0)$ cancels from the mixing free energy in
Eq.~\eqref{eq:fmix}.

\paragraph{The neural ladder.}
The trained sampler generates trajectories at fixed composition. A terminal
configuration on rung $n$ supplies the A$\to$B substitution works
$w_{\mathrm f}$ for the edge $n\to n+1$ and the B$\to$A substitution works
$w_{\mathrm r}$ for the edge $n-1\to n$ (the pure rungs $n=0$ and $n=N$
contribute to one edge only). All eligible substitutions in each direction are evaluated in one batched call to the potential, in addition to the terminal target-density evaluation required for the forward--backward path weight.

Because generated terminals are not assumed to follow the exact rung
ensemble, each trajectory is corrected by its forward--backward weight from
Eq.~\eqref{eq:fbrnd}. For $K_n$ trajectories on rung $n$, with terminal
states $x_k^{(n)}$ and log forward--backward path weights $\log W_k^{(n)}$, the normalized
within-rung weights and expectations under the exact rung ensemble read
\begin{equation}
\bar W_k^{(n)}
=
\frac{\exp\!\big(\log W_k^{(n)}\big)}
{\sum_{\ell=1}^{K_n}\exp\!\big(\log W_\ell^{(n)}\big)}, \qquad\quad \widehat{\langle O\rangle}_n
=
\sum_{k=1}^{K_n}
\bar W_k^{(n)} O\big(x_k^{(n)}\big).
\label{eq:si-nlad-rw-expectation}
\end{equation}

Thus, the one-sided estimators in Eq.~\eqref{eq:ladder} are evaluated by
applying Eq.~\eqref{eq:si-nlad-rw-expectation} to their corresponding
site-averaged exponential works.

The same correction is applied to BAR. For the edge $n\to n+1$, let
$n_{\mathrm f}=K_n$, $n_{\mathrm r}=K_{n+1}$, and
$C=\log(n_{\mathrm r}/n_{\mathrm f})$. The importance-weighted BAR equation is
\begin{equation}
\begin{aligned}
&n_{\mathrm f}
\sum_{k=1}^{n_{\mathrm f}}
\bar W_k^{(n)}
\frac{1}{N-n}
\sum_{i:\,a_{k,i}^{(n)}=\mathrm A}
\mathrm{sigmoid}\!\left(
\varphi_n+C-w_{\mathrm f,i}^{(k)}
\right)
+
n_{\mathrm r}
\sum_{\ell=1}^{n_{\mathrm r}}
\bar W_\ell^{(n+1)}
\frac{1}{n+1}
\sum_{j:\,a_{\ell,j}^{(n+1)}=\mathrm B}
\mathrm{sigmoid}\!\left(
\varphi_n+C-w_{\mathrm r,j}^{(\ell)}
\right)
=
n_{\mathrm r}.
\end{aligned}
\label{eq:si-neural-bar}
\end{equation}
When all path weights are equal, $\bar W_k^{(n)}=1/n_{\mathrm f}$ and
$\bar W_\ell^{(n+1)}=1/n_{\mathrm r}$, and
Eq.~\eqref{eq:si-neural-bar} reduces to the unweighted BAR equation
\eqref{eq:si-bar}.

Any factor in $W_k^{(n)}$ that is constant across a rung cancels upon normalization. In particular, the semi-grand tilt $\exp(\beta\Delta\mu n)$ is common to all trajectories at fixed $n$ and therefore does not alter the target rung averages. 
The conditioning can nevertheless affect the efficiency of the proposal and hence the variance of the estimator, which we monitor through the within-rung effective sample size $\operatorname{ESS}_n=1/\sum_{k=1}^{K_n}\big(\bar W_k^{(n)}\big)^2$.

\subsection{Absolute free energies from Frenkel--Ladd integration}
\label{si:frenkel-ladd}
The free-energy methods described above determine mixing and relative free energies within a given crystal phase. To compare competing phases with different crystal structures, as required for the fcc--bcc phase competition in Ni--Cr, their free energies must additionally be placed on a common absolute scale. 
We obtain these absolute references from the pure-element Ni-fcc, Cr-fcc, Ni-bcc and Cr-bcc crystals using the Frenkel--Ladd method \cite{frenkel1984new}. In this approach, each crystal is coupled to an Einstein reference in which every atom is harmonically tethered to its corresponding lattice site,
\begin{equation}
U_{\mathrm{spring}}=\frac{1}{2}k\sum_i\left|\mathbf{r}_i-\mathbf{r}_i^0\right|^2,
\end{equation}
where $\mathbf{r}_i$ are the Cartesian atomic positions and $\mathbf{r}_i^0$ the corresponding lattice sites. The Einstein reference and target system are connected through $U(\lambda)=(1-\lambda)U_{\mathrm{spring}}+\lambda U_{\mathrm{target}}$, where $U_{\mathrm{target}}$ denotes the target interatomic potential. 
Thermodynamic integration then gives the classical Helmholtz free energy per atom as
\begin{equation}
f_{\mathrm{target}}=f_{\mathrm{Einstein}}+\frac{1}{N}\int_0^1\left\langle U_{\mathrm{target}}-U_{\mathrm{spring}}\right\rangle_\lambda\,d\lambda,
\end{equation}
with
\begin{equation}
f_{\mathrm{Einstein}}=3k_{\mathrm B}T\log\left(\frac{\hslash\omega_E}{k_{\mathrm B}T}\right),\qquad\omega_E=\sqrt{\frac{k}{m}}.
\end{equation}
The spring constant is chosen as the mean Cartesian force constant, $k=\mathrm{tr}(H)/(3N)$, obtained from the Hessian $H$ of the perfect crystal. 
Canonical averages at each coupling parameter are evaluated using Langevin molecular dynamics, with the cell volume independently relaxed at each temperature under zero external pressure, such that the resulting equilibrium Helmholtz free energy is equivalent to the Gibbs free energy per atom at $P=0$.
The resulting pure-element absolute free energies provide the reference offsets required to place the fcc and bcc free-energy curves on the same scale and directly determine their relative thermodynamic stability.

\subsection{Reference data generation}
\label{si:ref-data}

Two Monte Carlo references cover the alloy systems, matching the two model routes of Sections~\ref{si:legendre} and~\ref{si:fixed-comp}: First, for Cu--Ni, which remains miscible across the studied window, the composition responds smoothly to the exchange potential and chains sampling the semi-grand target \eqref{eq:target} directly converge with local moves alone; the reference is a grid of independent $(T,\Delta\mu)$ states pooled by multistate reweighting. Second, for Cu--Ag and Ni--Cr the miscibility gap makes that route unreliable as semi-grand chains are hysteretic at coexistence. The reference free energies and binodals for these systems are therefore computed with the canonical composition ladder---fixed-composition ensembles at every rung $n=0,\dots,N$, connected by the substitution-edge estimator of Section~\ref{si:ladder-bar}---which measures $\log Z_n$ rung by rung and never has to cross the barrier at all.

\paragraph{Local moves.} Both reference samplers run the same Metropolis kernel, batched across all states and walkers; only the species move differs. One sweep applies three moves. (i)~A collective displacement move $\mathbf u'=\mathbf u+\delta_u\boldsymbol\xi$, $\boldsymbol\xi\sim\mathcal N(0,\mathbf I)$; the proposal is symmetric and leaves $(v,\mathbf a)$ untouched, so $\log A=-\beta\,\Delta U$. (ii)~A log-volume move $v'=v+\delta_v\xi$, which does carry a Jacobian because $v$ changes, $\log A=-\beta\big(\Delta U+P\,\Delta V\big)+N\,\Delta v$; proposing in $\log V$ rather than $V$ turns the $V^N$ factor of the target into the harmless additive $N\Delta v$ and lets a fixed-width Gaussian proposal work across the whole grid. (iii)~The species move: the semi-grand chains attempt $n_{\mathrm{flip}}$ single-site species flips at uniformly random sites, accepted with $\log A=\beta\big(\Delta\mu\,\Delta N_{\mathrm B}-\Delta U_i\big)$ while the ladder chains attempt $n_{\mathrm{exch}}$ pair exchanges, in which a uniformly drawn A--B site pair trades types with $\log A=-\beta\,\Delta U$, rearranging the configuration at exactly fixed $N_{\mathrm B}$. One sweep costs $2+n_{\mathrm{flip}}$ (respectively $2+n_{\mathrm{exch}}$) batched potential evaluations. The displacement and volume moves act on the whole configuration at once, so a healthy step size scales as $\delta\sim1/\sqrt{N\beta}$ and cannot be a single constant across several hundred kelvin: steps are initialized per state as $\delta_k\propto\sqrt{\beta_{\min}/\beta_k}$ and adapted toward a target acceptance of ${\approx}0.3$ during burn-in only. At the first production sweep the steps freeze and all counters reset, so recorded samples come from a fixed, detailed-balanced kernel.

\paragraph{The semi-grand grid (Cu--Ni).} A reference run is a two-dimensional grid of $(T,\Delta\mu)$ states, each with two walkers initialized from the two pure phases (all-Ni and all-Cu), whose agreement after burn-in is the convergence criterion: the two chains approach the stationary composition from opposite sides. The states are statistically independent and no configurations are exchanged between them, so a grid shards freely across devices, and the states are connected only in post-processing, by the multistate reweighting below.

\paragraph{Multistate reweighting.}
Samples from $K$ semi-grand states $\{(T_k,\Delta\mu_k,P_k)\}$ are combined with the multistate Bennett acceptance ratio (MBAR) \cite{shirts2008statistically}. Each pooled configuration $x=(\mathbf a,\mathbf u,v)$ is scored under every state through its reduced potential,
\begin{equation}
u_k(x)=\frac{U(x)+P_k V(x)-\Delta\mu_k\,N_{\mathrm B}(x)}{k_{\mathrm B}T_k},
\label{eq:reduced}
\end{equation}
which is $-\log\pi_1(T_k,\Delta\mu_k)$ up to the Jacobian term $Nv$; the latter is state-independent and cancels from all weight ratios. With $N_k$ samples drawn from state $k$, the reduced free energies $\hat f_k=-\log\Xi_k$ (up to one common constant) solve the self-consistent MBAR equations
\begin{equation}
\hat f_k=-\log\sum_{x}\frac{e^{-u_k(x)}}{\sum_{k'}N_{k'}\,e^{\hat f_{k'}-u_{k'}(x)}},
\label{eq:mbar}
\end{equation}
where the outer sum runs over the pooled samples of all states. Since MBAR assumes decorrelated samples, the pooled chains are thinned by their measured autocorrelation.

\paragraph{The composition ladder (Cu--Ag and Ni--Cr).} We use two walkers per $(n,T)$ state, one initialized as a random solid solution, the other as a phase-separated slab, and use the residual spread between them as convergence criterion. Within each rung, adjacent temperatures exchange full configurations in disjoint pairs of alternating parity (the deterministic even--odd scheme \cite{syed2022nonreversible}). Both exchange partners share $N_{\mathrm B}$, so the composition term of the target cancels from the acceptance,
\begin{equation}
\log A_{\mathrm{swap}}=(\beta_k-\beta_l)\big[\mathcal H(x_k)-\mathcal H(x_l)\big],\qquad \mathcal H=U+PV,
\label{eq:si-swap}
\end{equation}
and the volume Jacobian $Nv$ cancels identically because its coefficient is the same $N$ in every replica.

\section{Experimental setup}
\label{si:setup}

\subsection{Systems and reference data}
\label{si:setup-ising}

\paragraph{The Ising testbed.} The lattice semi-grand-canonical alloy is the controlled testbed for the species channel where particles sit on an $L\times L$ periodic lattice. Starting from the standard ferromagnetic Ising Hamiltonian,
\begin{equation}
H_{\mathrm{Ising}}(\boldsymbol{\sigma}) = -J\sum_{\langle ij\rangle}\sigma_i\sigma_j - h\sum_i\sigma_i, \qquad \sigma_i \in \{-1, +1\},
\label{eq:ising-standard}
\end{equation}
we substitute $\sigma_i = 2a_i - 1$ to map spins to species occupancy labels $a_i \in \{0, 1\}$. This yields the semi-grand-canonical alloy Hamiltonian
\begin{equation}
H(\mathbf a) = \sum_{\langle ij\rangle}\epsilon(a_i,a_j) - \Delta\mu\,N_{\mathrm B}(\mathbf a), \qquad \pi_1(\mathbf a) \propto e^{-H(\mathbf a)/k_{\mathrm B}T},
\label{eq:si-ising-H}
\end{equation}
where $N_{\mathrm B}(\mathbf a) = \sum_i a_i$. Setting pair energies $\epsilon_{AA}=\epsilon_{BB}=0$ and $\epsilon_{AB}=2$ directly matches the Ising parameters via $J=(2\epsilon_{AB}-\epsilon_{AA}-\epsilon_{BB})/4 = 1$ and uniform field $h = \Delta\mu/2$, with critical temperature $(k_{\mathrm B}T/J)_c = 2/\log(1+\sqrt2) \approx 2.2692$. The $(T,\Delta\mu)$ plane is therefore the textbook field--temperature phase diagram: a first-order coexistence line along $\Delta\mu=0$ for $T<T_c$, terminating at the critical point $(T_c,0)$, and analytic crossover elsewhere. A single amortized model is trained per lattice size and read out over the whole plane; the results of Fig.~\ref{fig2}a use $L=16$. Ground truth is generated with a \emph{ghost-spin} Wolff cluster sampler \cite{wolff1989collective} on a grid of 11 temperatures $k_{\mathrm B}T/J\in[1.5,3.2]$ (containing $T_c$) and 11 fields $\Delta\mu\in[-0.2,0.2]$, the $\Delta\mu<0$ half-plane following by particle--hole symmetry; each state runs 24 independent chains of 3\,000 cluster steps (600 discarded as burn-in).

\paragraph{Alloy systems.}
\label{si:setup-alloys}
The alloy systems are those considered in the main part: fcc Cu--Ni, Cu--Ag and Ni--Cr on $3\times3\times3$ conventional supercells ($N=108$ sites) and bcc Ni--Cr on a $4\times4\times4$ supercell ($N=128$), at $P=0$, with EAM potentials from the NIST repository \cite{hale2018evaluating}---Cu--Ni \cite{fischer2019systematic}, Cu--Ag \cite{williams2006embedded}, and a Finnis--Sinclair-type Ni--Co--Cr potential \cite{li2019strengthening} for both Ni--Cr lattices---and interaction cutoffs of 5.0\,\AA{} (Cu--Ni, Ni--Cr fcc) and 5.3\,\AA{} (Cu--Ag, Ni--Cr bcc). The amortization windows are 600--1200\,K (Cu--Ni), 500--1000\,K (Cu--Ag) and 600--1500\,K (both Ni--Cr lattices) in temperature; the chemical-potential windows are specified below (Table~\ref{tab:si-windows}).

\paragraph{Alloy reference data.} All semi-grand references follow the local-move protocol of Section~\ref{si:ref-data}, with two walkers per state initialized from the two competing pure phases; their agreement after burn-in is the convergence certificate. Cu--Ni: 15 temperatures (500--1200\,K) $\times$ 33 chemical potentials ($\Delta\mu=0.600$--$1.150$\,eV, 10\,meV spacing across the composition crossover and 25\,meV in the pinned wings), 40\,000 sweeps (2\,000 burn-in) with 6 species-flip attempts per sweep and no exchange between states (local moves only). The $N=256$ size-transfer reference of Fig.~\ref{fig2}e uses 7 temperatures $\times$ 11 chemical potentials with temperature exchange every other sweep and 24\,000 sweeps (12\,000 burn-in).
The composition-ladder references of Section~\ref{si:ladder-bar} sample every rung $n=0,\dots,N$---at 11 isotherms (500--1000\,K) for Cu--Ag and 7 isotherms (600--1500\,K) per Ni--Cr lattice---with two walkers per rung initialized from a random solid solution and a phase-separated slab (their spread is the per-edge convergence gate), 30\,000 sweeps (5\,000 burn-in) of displacement, volume and pair-exchange moves, and temperature replica exchange at fixed composition.

\subsection{Architecture and training details}
\label{si:training-details}

\paragraph{Alloy network.} All heads share one PaiNN-style equivariant message-passing trunk \cite{schutt2021equivariant} with 4 interaction layers and 64 scalar and 64 vector channels per site ($2.7\times10^5$ parameters); interatomic distances are expanded in 16 Gaussian radial basis functions under a smooth polynomial cutoff envelope, with the cutoff matched to the potential range as above. The per-site input features encode the species state, smooth species-resolved coordination numbers together with a separate count of still-masked neighbors, the displacement magnitude, the normalized log-volume, sinusoidal time features, and the conditioning features of the amortized family---the normalized inverse temperature, the normalized offset $\Delta\mu-\widehat{\Delta\mu}(T)$. The displacement velocity and score are emitted by per-site gated equivariant vector heads, augmented by a learned scalar multiple of the displacement itself---a lattice anchor that breaks exactly the symmetry the reference lattice breaks; the volume fields are read out from mean-pooled invariant features; and the species head emits per-site categorical logits from layer-normalized, gradient-isolated trunk features, so that the $O(1)$ cross-entropy \eqref{eq:loss-disc} is not swamped by the large-scale score regressions. All scalar heads are zero-initialized, so generation starts from the prior transport and the fair-coin completion law.

\paragraph{Lattice network.} For the Ising testbed the species head is a periodic fully-convolutional network: 4--6 layers of $3\times3$ convolutions with circular padding and 64--96 channels ($0.2$--$0.8$\,M parameters), with time, the conditions and the spatial mean of every feature channel re-injected at each layer as constant input channels---an intensive global feature that gives the local receptive field access to the lattice-wide order parameter---and a zero-initialized $1\times1$ output layer.

\paragraph{Optimization and budget.} All models train with Adam at learning rate $3\times10^{-3}$; for the alloys the species loss weight in Eq.~\eqref{eq:cond-loss} is $\lambda=2$ and the two continuous channels are weighted equally (Cu--Ag triples the volume-channel weight). Rollouts use $M=100$ integration steps for the alloys and 128 reveal blocks for the lattice models during training. Alloy training runs 100--120 outer rounds, each generating 500 fresh terminals (1\,000 for Cu--Ni) into a replay buffer of 2\,000--5\,000 configurations and taking 500 gradient steps at batch size 64--96. Because rollouts are target-free and the labels are buffered (Section~\ref{si:algorithms}), the total number of potential evaluations per training run is the initial buffer fill plus rounds $\times$ fresh terminals: $1.1\times10^{5}$ (Cu--Ag), $1.3\times10^{5}$ (per Ni--Cr lattice) and $1.25\times10^5$ (Cu--Ni), each evaluation a single batched energy--force--virial--substitution pass. For the lattice testbed the Hamiltonian is inexpensive and the heat-bath labels are instead recomputed on the fly; training runs several hundred outer rounds with 512 chains per generation.

\begin{table}[tbp]
\centering
\caption{Amortization windows and transition-line estimates $\widehat{\Delta\mu}(T)=a_0+a_1T$ of the alloy models.}
\label{tab:si-windows}
\begin{tabular}{lcccc}
\hline
system & $T$ window (K) & $a_0$ (eV) & $a_1$ (eV/K) & narrow half-width (eV) \\
\hline
Cu--Ni & 600--1200 & $+0.893$ & $-5.4\times10^{-5}$ & 0.06 \\
Cu--Ag & 500--1000 & $+0.690$ & $-7.1\times10^{-5}$ & 0.15 \\
Ni--Cr fcc & 600--1500 & $-0.909$ & $-5.4\times10^{-5}$ & 0.06 \\
Ni--Cr bcc & 600--1500 & $-1.086$ & $+0.9\times10^{-5}$ & 0.06 \\
\hline
\end{tabular}
\end{table}
\paragraph{Choice of the conditioning distribution.}
\label{si:conditioning}
Temperature is sampled uniformly in the inverse temperature over the training window; since the target log-density \eqref{eq:target} is linear in $(k_{\mathrm B}T)^{-1}$, this covers the range of score magnitudes evenly. The chemical potential is sampled relative to an inexpensive linear estimate of the transition line, $\Delta\mu=\widehat{\Delta\mu}(T)+\delta$ with $\widehat{\Delta\mu}(T)=a_0+a_1T$, whose anchor $a_0$ is the $T=0$ energy difference of the relaxed pure phases and whose slope $a_1$ comes from a short pilot run; away from this line the composition is pinned at the pure phases and the target barely changes with $\Delta\mu$. The offset $\delta$ is drawn from an equal mixture of a wide uniform component ($\pm0.30$\,eV), which covers the composition-pinned wings and the estimation error, and a narrow one ($\pm0.06$\,eV; $\pm0.15$\,eV for Cu--Ag), which concentrates the budget where the composition responds most strongly. The estimate needs only to be approximately correct: it determines what the conditioned prior absorbs, and the network corrects the remainder. The per-system windows and line estimates are listed in Table~\ref{tab:si-windows}.

\paragraph{Physics-informed priors.}
\label{si:prior-fit}
 The volume prior is conditioned on $(T,\Delta\mu)$ through the composition estimate $c_0=\mathrm{sigmoid}\big((\Delta\mu-\widehat{\Delta\mu}(T))/k_{\mathrm B}T\big)$---the exact isotherm of the non-interacting lattice gas at the inexpensive linear estimate $\widehat{\Delta\mu}(T)$ of the phase-transition line---and the displacement prior on $T$ alone:
\begin{align}
v_0&\sim\mathcal N(\bar v,\sigma_v^2), & e^{\bar v}&=N\,\hat V_{\mathrm{atom}}(c_0,T),\label{eq:prior-w}\\
\mathbf u_0&\sim\mathcal N\big(0,\sigma_u(T)^2 I\big), & \sigma_u(T)&=\sigma_u^{\mathrm{ref}}\sqrt{T/T_{\mathrm{ref}}}.\label{eq:prior-u}
\end{align}
The volume prior \eqref{eq:prior-w} centers the cell volume on the per-atom volume
\begin{equation}
\hat V_{\mathrm{atom}}(c_0,T)=\big[(1-c_0)\,\hat V_{\mathrm A}+c_0\,\hat V_{\mathrm B}+\Omega_V\,c_0(1-c_0)\big]\,\big(1+\alpha_V\,(T-T_{\mathrm{ref}})\big),
\label{eq:si-vhat}
\end{equation}
an interpolation between the per-atom volumes $\hat V_{\mathrm A},\hat V_{\mathrm B}$ of the relaxed pure phases, a regular-solution excess term $\Omega_V\,c_0(1-c_0)$, and a linear thermal-expansion factor. All four parameters are derived from the interatomic potential alone, so no reference data enters training. Two stages generate the fitting data. \emph{Statics}: on a composition grid $c\in\{0,\tfrac14,\tfrac12,\tfrac34,1\}$, displacements and log-volume are relaxed by gradient descent on the energy at $P=0$ averaged over random species arrangements; the local relaxation around size-mismatched neighbors is what produces the excess volume, so unrelaxed lattices would spuriously fit $\Omega_V\approx0$. \emph{Finite temperature}: short fixed-composition isobaric Metropolis runs on the same energy (displacement and volume moves only, so the composition never changes) over a $(c,T)$ grid supply the thermal volume surface. The fit exploits that \eqref{eq:si-vhat} is nonlinear only through the scalar $\alpha_V$: at fixed $\alpha_V$ the thermal factor $s(T)=1+\alpha_V(T-T_{\mathrm{ref}})$ is a known number for every data point, so each measured volume contributes one row of the linear model
\begin{equation}
V_{\mathrm{atom}}(c,T)\;\approx\;s(T)\,\big[(1-c),\;c,\;c(1-c)\big]\cdot\big(\hat V_{\mathrm A},\,\hat V_{\mathrm B},\,\Omega_V\big)^{\!\top},
\label{eq:si-vfit}
\end{equation}
and the three linear parameters follow from an exact normal-equation least-squares solve. A one-dimensional scan over $\alpha_V$, keeping the solution with the smallest residual, then determines all four parameters.

The displacement-prior width $\sigma_u^\text{ref}$ is obtained on the same energy-only budget: at each point of a small $(c,T)$ grid the displacements are relaxed at the conditioned volume $\hat V_{\mathrm{atom}}(c,T)$, one Hessian $K=\partial^2U/\partial\mathbf u^2$ of the energy is evaluated in fractional coordinates, and the isotropic width of the quasi-harmonic law $\mathbf u_0\sim\mathcal N\big(0,(\beta K)^{-1}\big)$, namely $\sigma_{\mathrm{eff}}(T)=\sqrt{\operatorname{tr}\big[(\beta K)^{-1}\big]/3N}$, is regressed as $\log\sigma_{\mathrm{eff}}(T)=\log\sigma_u^{\mathrm{ref}}+p\,\log(T/T_{\mathrm{ref}})$, where $\sigma_u^\text{ref}$ and $p$ are the fitting parameters. The entire prior calibration costs only a few energy evaluations per system.


\bibliography{refs.bib}
\end{document}